\documentclass[12 pt]{article}
\usepackage{bbm}
\usepackage{amssymb}
\usepackage{amsmath}
\usepackage{amsthm}
\usepackage[onehalfspacing]{setspace}
\usepackage{sectsty}
\usepackage{titlesec}
\usepackage{hyperref}
\usepackage{graphicx}
\usepackage{caption}
\usepackage{subcaption}
\usepackage[shortlabels]{enumitem}
\usepackage{fullpage}
\usepackage{lmodern}
\usepackage{comment}
\usepackage[font={small,it},justification=justified]{caption}
\usepackage[disable]{todonotes}
\usepackage{thmtools}
\usepackage{thm-restate}

\setlist{noitemsep,parsep=6pt,partopsep=0pt,topsep=0pt}
\usepackage[margin=1in]{geometry}
\usepackage{verbatim}

\usepackage[
   backend=biber,
   backref=false,
    bibstyle=authoryear,
    sortlocale=en_US,
    url=false, 
    doi=false,
    eprint=false,
		isbn=false,
maxbibnames=6,
maxcitenames=3,
date=year,
bibstyle=authoryear,
citestyle=authoryear,
natbib
]{biblatex}

\renewbibmacro{in:}{%
 \ifentrytype{article}{, }{%
 \printtext{\bibstring{in}\intitlepunct}}}

\DeclareFieldFormat[article,book,inbook,incollection,inproceedings,patent,thesis,unpublished]{title}{#1}

\DeclareDelimFormat{nameyeardelim}{\addcomma\space}

 \theoremstyle{remark}

\theoremstyle{plain}
\newtheorem{theorem}{Theorem}

\newtheorem{assumption}{Assumption}

\newtheorem{corollary}{Corollary}

\newtheorem{example}{Example}

\newtheorem{lemma}{Lemma}

\newtheorem{proposition}{Proposition}

\renewcommand{\epsilon}{\varepsilon}

\newcommand{\R}{\mathbb{R}}
\newcommand{\E}{\mathbb{E}}
\newcommand{\Prob}{\mathbb{P}}
\newcommand{\bbeta}{\boldsymbol{\beta}}

\newcommand{\U}{\mathcal{U}}
\newcommand{\M}{\mathcal{M}}
\newcommand{\C}{\mathcal{C}}
\newcommand{\V}{\mathcal{V}}
\newcommand{\Q}{\mathcal{Q}}
\renewcommand{\H}{\mathcal{H}}
\newcommand{\n}{\hat{\mathbf{n}}}
\newcommand{\tU}{\tilde{U}}
\DeclareMathOperator{\diver}{div}

\DeclareMathOperator{\bd}{bd}

\DeclareMathOperator{\interior}{int}
\DeclareMathOperator{\conv}{co}

\DeclareMathOperator{\supp}{supp}

\newcommand{\norm}[1]{ \lVert #1 \rVert }

\makeatletter
  \renewcommand\@seccntformat[1]{\csname the#1\endcsname.{\hskip.7em\relax}} %
\makeatother
\makeatletter
\renewenvironment{proof}[1][\proofname] {\par\pushQED{\qed}\normalfont\topsep6\p@\@plus6\p@\relax\trivlist\item[\hskip\labelsep\bfseries#1\@addpunct{.}]\ignorespaces}{\popQED\endtrivlist\@endpefalse}
\makeatother

\newcommand{\mailto}[1]{\href{mailto:#1}{\texttt{#1}}} %

\let\oldfootnote\footnote
\renewcommand\footnote[1]{\oldfootnote{\hspace{.5mm}#1}}

\newcommand{\appendixref}[1]{\hyperref[#1]{Appendix \ref{#1}}}

\usepackage{xcolor}
\usepackage{pgf}
\usepackage{tikz}
\usetikzlibrary{shapes,arrows,arrows.meta,automata,fit}

\definecolor{hcol}{RGB}{38,94,148}   %
\definecolor{ucol}{RGB}{20,112,88}   %
\definecolor{viol}{RGB}{176,66,48}   %

\tikzset{
  hcurve/.style ={hcol, line width=0.9pt},
  ucurve/.style ={ucol, line width=1.5pt},
  uext/.style   ={ucol, line width=1.0pt, dash pattern=on 3pt off 2.2pt},
  utilta/.style ={ucol, line width=1.1pt, dash pattern=on 4pt off 2.5pt},
  utiltb/.style ={ucol, line width=1.1pt, dash pattern=on 1pt off 2pt},
  guide/.style  ={black!40, line width=0.4pt, dash pattern=on 1.2pt off 1.6pt},
  saxis/.style  ={black!45, line width=0.5pt},
  sint/.style   ={viol, line width=1.7pt},
  kink/.style   ={ucol, fill=white, line width=0.8pt},
  ann/.style    ={font=\footnotesize, inner sep=1.5pt}
}

\tikzset{
  menu/.style  ={hcol, line width=1.4pt},
  bench/.style ={viol, line width=1.1pt, dash pattern=on 3pt off 2.2pt},
  pool/.style  ={black!45, line width=0.4pt, dash pattern=on 1.4pt off 1.6pt},
  sq/.style    ={black!70, line width=0.6pt}
}
\newcommand{\unifcurve}{plot[smooth] coordinates {(-0.2044,0.7956) (-0.1533,0.7454) (-0.1022,0.6974) (-0.0511,0.6521) (0.0000,0.6103) (0.0490,0.5740) (0.1025,0.5377) (0.1587,0.5014) (0.2150,0.4651) (0.2683,0.4288) (0.3159,0.3925) (0.3563,0.3563) (0.3925,0.3159) (0.4288,0.2683) (0.4651,0.2150) (0.5014,0.1587) (0.5377,0.1025) (0.5740,0.0490) (0.6103,0.0000) (0.6521,-0.0511) (0.6974,-0.1022) (0.7454,-0.1533) (0.7956,-0.2044)}}
\newcommand{\unifregion}{plot[smooth] coordinates {(0.0000,0.6103) (0.0490,0.5740) (0.1025,0.5377) (0.1587,0.5014) (0.2150,0.4651) (0.2683,0.4288) (0.3159,0.3925) (0.3563,0.3563) (0.3925,0.3159) (0.4288,0.2683) (0.4651,0.2150) (0.5014,0.1587) (0.5377,0.1025) (0.5740,0.0490) (0.6103,0.0000)} -- (0,0) -- cycle}
\newcommand{\unifaxes}{%
\draw[sq] (0,0) rectangle (1,1);
\draw[saxis,->] (-0.80,0) -- (1.14,0) node[ann,right]{$s_1$};
\draw[saxis,->] (0,-0.80) -- (0,1.14) node[ann,above]{$s_2$};}
\newcommand{\unifpanellarge}[1]{%
\begin{tikzpicture}[x=#1,y=#1]
\begin{scope}
\clip (-0.80,-0.80) rectangle (1.12,1.12);
  \foreach \d in {-0.6,-0.3,0,0.3,0.6}{
    \pgfmathsetmacro\xa{max(\d,0)} \pgfmathsetmacro\ya{max(-\d,0)}
    \pgfmathsetmacro\xb{min(1,1+\d)} \pgfmathsetmacro\yb{min(1,1-\d)}
    \draw[pool] (\xa,\ya) -- (\xb,\yb);}
  \draw[menu] (-0.6667,0.3333) -- (0.3333,-0.6667);
  \foreach \d in {-0.6,-0.3,0.3,0.6}{\draw[hcol,fill=white,line width=0.7pt] ({\d/2-0.1667},{-\d/2-0.1667}) circle (1.5pt);}
  \draw[viol,line width=1.1pt] (-0.2017,-0.2017) -- (-0.1317,-0.1317);
  \draw[viol,line width=1.1pt] (-0.2017,-0.1317) -- (-0.1317,-0.2017);
\end{scope}
\unifaxes
\end{tikzpicture}}
\newcommand{\unifpanelsmall}[1]{%
\begin{tikzpicture}[x=#1,y=#1]
\begin{scope}
\clip (-0.80,-0.80) rectangle (1.12,1.12);
  \fill[hcol!12] \unifregion;
  \foreach \a/\b/\c/\d in {0.0656/0.5625/0.3662/1.0, 0.2303/0.4550/0.5944/1.0, -0.1022/0.6974/0.1673/1.0}
     \draw[pool] (\a,\b) -- (\c,\d);
  \foreach \a/\b/\c/\d in {0.5625/0.0656/1.0/0.3662, 0.4550/0.2303/1.0/0.5944, 0.6974/-0.1022/1.0/0.1673}
     \draw[pool] (\a,\b) -- (\c,\d);
  \draw[pool] (0.3563,0.3563) -- (1.0,1.0);
  \draw[bench] (0,0.5) -- (0.5,0.5) -- (0.5,0);
  \draw[menu] \unifcurve;
\end{scope}
\unifaxes
\end{tikzpicture}}

\tikzstyle{info}=[circle,thick,draw=black,fill=black!25,minimum size=4mm]
\tikzstyle{uninfo}=[circle,thick,draw=black,fill=white,minimum size=4mm]
\tikzstyle{inforecog}=[circle,line width=1mm,draw=black!50,fill=black!25,minimum size=4mm]
\tikzstyle{uninforecog}=[circle,line width=1mm,draw=black!50,fill=white,minimum size=4mm]

\tikzstyle{traded}=[draw, line width=1mm]
\tikzstyle{recog}=[draw=black!50, line width=1mm]

\titlespacing*{\paragraph}{0pt}{1.2ex plus .3ex minus .2ex}{1em}
\titlespacing*{\section}{0pt}{2.0ex plus 1ex minus .2ex}{1.2ex plus .2ex}
\titlespacing*{\subsection}{0pt}{1.8ex plus 1ex minus .2ex}{1.0ex plus .2ex}
\makeatletter\def\thm@space@setup{\thm@preskip=6pt \thm@postskip=6pt}\makeatother
\allowdisplaybreaks
\AtBeginDocument{\setlength{\abovedisplayskip}{6pt plus 2pt minus 2pt}
  \setlength{\belowdisplayskip}{6pt plus 2pt minus 2pt}
  \setlength{\abovedisplayshortskip}{3pt plus 2pt}
  \setlength{\belowdisplayshortskip}{4pt plus 2pt minus 2pt}}

\begin{document}

\begin{titlepage}

\title{Optimal Delegation in a Multidimensional World
}
\author{
Andreas Kleiner\thanks{Department of Economics, University of Bonn.  Email:  \mailto{andreas.kleiner@uni-bonn.de}. I thank Ricardo Alonso, Kyle Bagwell, Alex Frankel, Nima Haghpanah, Ravi Jagadeesan, Navin Kartik, Elliot Lipnowski, Alejandro Manelli, Lina Uhe, Kun Zhang, and various seminar audiences for useful comments and discussions. I acknowledge financial support from the German Research Foundation (DFG) through Germany’s Excellence Strategy - EXC 2047/1-390685813, EXC 2126/1-390838866, and the CRC TR-224 (Project B02).}
}
\date{ %
September 2026}

\maketitle
We study a model of delegation in which a principal takes a multidimensional action and an agent has private information about a multidimensional state. We provide necessary and sufficient conditions for an arbitrary mechanism to maximize the principal's expected payoff. A key step shows that a function is the indirect utility of an incentive-compatible mechanism if and only if it is convex and lies below the agent's full-discretion payoff. We apply our optimality conditions to convex delegation sets and characterize when a single budget optimally links decisions. Linking can preserve discretion even when separate delegation would eliminate it. Optimal linking can also require a nonconvex delegation set, even with uniformly distributed types. If the principal cares little about aligned decisions, every optimal mechanism randomizes: risk in an aligned decision provides incentives in the biased ones at little cost to her.

\thispagestyle{empty}

\end{titlepage}
\section{Introduction}

In many economic and political environments, a principal faces a better-informed but biased agent. The principal can choose a permissible set of actions and `delegate' the decision to the agent: a firm appoints a manager to select investment levels in different projects; 
the US Congress delegates power to federal agencies;
a legislature forms a committee to draft bills; a regulator lets a monopolist choose prices. The principal may yield some discretion to the agent to utilize his informational advantage but may impose restrictions on the agent's actions to counter his bias.
Following \citet{Holmstrom:77}, an extensive literature models such delegation problems by assuming that both the action and the state of the world lie in a one-dimensional space. A main result of this literature characterizes when it is optimal for the principal to constrain the agent's choice to lie in an interval, and this conclusion has been used to explain why managers face spending caps, regulators impose price ceilings, and trade agreements specify maximum tariff levels.

The assumption that the action and state spaces are one-dimensional is made for tractability. In many applications, the underlying states and actions are more complex and more realistically modeled as multidimensional: managers invest in several projects, Congress delegates many decisions to the EPA, and committees draft multiple bills. What mechanisms are optimal in such multidimensional settings? How robust are conclusions obtained for one-dimensional models? And can we still expect that relatively simple mechanisms are often optimal?

To study these questions, we consider a principal who takes a multidimensional action and faces an agent with private information about a multidimensional state of the world (the agent's \textit{type}). Payoffs depend on the action and the state of the world, and transfers are infeasible. The principal can design arbitrary mechanisms, including stochastic ones, to maximize her expected payoff. Our main result characterizes, for an arbitrary mechanism, when this mechanism is optimal. While our primary aim is to understand multidimensional models, our approach yields novel insights already for one-dimensional models because our conditions apply to any mechanism. For multidimensional convex delegation sets we obtain a simple optimality condition that naturally extends familiar conditions for interval delegation in one-dimensional models. We also obtain conditions under which linking all decisions through a single budget is optimal. As extensions, we allow more general payoffs for the principal and an agent with an outside option, covering veto bargaining and monopoly regulation.

To illustrate, consider the following example, to which we return throughout the paper. The state has two independent components, each uniformly distributed on $[0,1]$, and payoffs are quadratic and additively separable across the two dimensions. The agent's ideal action is the state, and the principal's ideal action is lower by a common bias $\beta>0$ in each dimension. If the two decisions had to be delegated separately, the one-dimensional theory would apply to each of them: in the optimal mechanism the agent may choose any action up to a cap. However, a powerful insight in mechanism design is that linking decisions can be useful even if they are independent \citep{jackson07}. 

Indeed, the optimal action cap with a single decision is so restrictive that in high states of the world, both the principal and the agent would prefer higher actions. But if the principal were to allow the agent to take higher actions, the agent would choose them too often and the principal would be better off with the restriction. With several dimensions, however, the principal can allow the agent to take a more extreme action in one dimension in exchange for a more moderate action in the other. How should she optimally link the two decisions? \autoref{fig:intro_example} shows that even in this uniform example, the optimal delegation sets can take different forms as the bias varies.

\begin{figure}[ht]
\centering
\begin{subfigure}[t]{0.47\textwidth}\centering
\unifpanellarge{2.5cm}
\end{subfigure}\hfill
\begin{subfigure}[t]{0.47\textwidth}\centering
\unifpanelsmall{2.5cm}
\end{subfigure}
\caption{Optimal delegation sets in the two-dimensional uniform example: left for $\beta=2/3$, right for $\beta=1/4$. Optimal separate delegation is shown in red.}
\label{fig:intro_example}
\label{fig:optimal_bundling}
\end{figure}

For every bias, it is optimal to delegate a set of actions to the agent, even though the principal could also use stochastic mechanisms, and our results determine the shape of that set. Moreover, linking remains valuable however large the bias: if $\beta\ge 1/2$, there would be no value to delegation in separate delegation problems, and the optimal delegation set would collapse to a single action. By bundling the two dimensions, the principal can still give the agent discretion, and the optimal delegation set is a line, as in the left panel for $\beta=2/3$. It is then optimal to fix the \emph{sum} of the two actions while leaving their \emph{difference} to the agent; \autoref{cor:symmetric_logconcave} shows that this is not special to the uniform distribution. For $\beta<1/2$, by contrast, the optimal set is nonconvex and has a curved boundary, as the right panel of \autoref{fig:intro_example} illustrates for $\beta=1/4$. The rate at which the agent may trade a higher action in one dimension against a lower action in the other varies along the boundary, so optimal linking is not simply a cap on the sum of the two decisions. 

An important step is understanding which mechanisms are feasible for the principal in the first place. 
We characterize the feasible mechanisms by analyzing the \textit{indirect utility} of a mechanism, which assigns to any type the payoff this type would get in the mechanism by choosing his report optimally. This payoff cannot exceed the \textit{full discretion payoff}, i.e., the payoff this type would get if he could choose the action without any restrictions. Moreover, our assumption that the agent's utility is an affine function of the state implies that the indirect utility must be a convex function. \autoref{lemma:feasible_set} shows that any function satisfying these two properties is the indirect utility of an incentive-compatible mechanism.

This characterization is easy to use and already helpful for one-dimensional delegation models. 
Our formulation differs from the previous literature, which often considered only deterministic mechanisms. The set of incentive-compatible deterministic mechanisms is not convex, since a convex combination of two such mechanisms is in general not incentive compatible.\footnote{Some earlier papers also consider stochastic mechanisms (or allow for money burning or restricted transfers) and obtain a convex set of mechanisms; see, for example, \citet{KM:09,AB:13,AE:17,amador20,kartik21,kleiner21}.} Moreover, a common approach is to first treat the model as one with transfers and then impose that these transfers are zero. Compared to this approach, formulating the problem via indirect utilities is more direct and provides valuable geometric insights into which mechanisms can be optimal. For the multidimensional problem, working with indirect utilities allows us to express incentive compatibility through convexity, building on \citet{rochet87}.

To find the optimal mechanism, we formulate the principal's problem in terms of indirect utilities. To illustrate and simplify the approach, we assume in \autoref{sec:optimal} that the principal's payoff has the same curvature as the agent's payoff (we extend the analysis to more general preferences for the principal in \autoref{sec:general_payoffs}). In this case, the problem becomes a linear program, and we use linear programming duality to derive necessary and sufficient conditions for a given mechanism to be optimal. Typically, optimal mechanisms pool certain types, and our main result shows that a mechanism is optimal if and only if, conditional on any pooling region, a stochastic dominance condition is satisfied. Intuitively, this condition requires that, restricted to the pooling region (where the indirect utility function is affine), no feasible indirect utility does better than the mechanism under consideration.
If the pooling regions are at most one-dimensional, the stochastic dominance condition has a simple formulation in terms of majorization. Using this observation, we provide necessary and sufficient conditions for a convex delegation set with a smooth boundary to be optimal when payoffs are quadratic. These conditions are easy to check and extend those of \citet{AB:13} to the multidimensional setting. We show that optimal mechanisms are deterministic if virtual values are positive in the interior of the type space and negative only on one of its faces. With more general preferences, however, every optimal mechanism randomizes if the principal cares little about decisions on which she agrees with the agent, since risk in these decisions provides incentives in the others. 

The rest of the paper proceeds as follows. After reviewing the related literature, we present our model in \autoref{sec:model}, characterize incentive-compatible mechanisms in \autoref{sec:ic}, and optimal mechanisms in \autoref{sec:optimal}. In \autoref{sec:general_payoffs} we consider more general preferences for the principal, and in \autoref{sec:outside_option} we incorporate outside options and apply them to veto bargaining and to the regulation of a monopolist.

\paragraph{Related Literature}

For the most part, the literature on delegation has focused on problems in which both the action and state spaces are one-dimensional; for important contributions, see \citet{Holmstrom:77,MS:91,AM:08,KM:09,AB:13,KZ:19,ABC:25}.

A few delegation papers do consider richer action and/or type spaces. \citet{Armstrong95} considers an agent with two-dimensional private information and discusses several applications. Since the principal's action is assumed to be one-dimensional (and only interval delegation sets are considered), there is only limited scope to screen two-dimensional types in his analysis. 
\citet{koessler2012optimal} characterize the optimal mechanism in a setting where two decisions depend on a single-dimensional underlying state.
\citet{galperti2019theory} studies personal budgeting for a consumer with self-control problems. He models the consumer's problem as a delegation problem with multidimensional information and actions but restricts the principal's choice to a particular class of ``budgeting mechanisms''. The closest paper to ours is \citet{frankel2016delegating}, which studies the delegation of several independent decisions and therefore has multidimensional action and state spaces. For quadratic preferences with a constant bias, he shows that if the states are independently and identically distributed according to normal distributions then it is optimal to delegate a `half space'. Without the normality assumption, he shows that suitable half-space delegation sets yield a payoff per decision converging to the first-best as the number of independent decisions grows, provided the state distributions have uniformly bounded variances.
\citet{frankel2014aligned} also considers multidimensional delegation problems and characterizes the max-min optimal mechanism, which maximizes the principal's payoff against the worst-case preference type of the agent. Recently, \citet{alonso26} show that convex delegation sets are robustly optimal if the designer faces uncertainty about the agent's preferences and evaluates mechanisms using a max-min criterion. 

The elicitation of multidimensional private information has also been studied in general mechanism design \citep[e.g.,][]{jackson07} and cheap-talk environments \citep{chakraborty2007comparative,lipnowski2020cheap}. \citet{jackson07} show that linking independent copies of a decision problem can yield expected payoffs per decision approaching those of an ex ante efficient allocation as the number of decisions grows. Our results can be used to study how the principal should optimally link decisions, which can be important if the number of decisions is limited.

On a methodological level, our work is related to the literature on multidimensional mechanism design, and in particular on multiproduct monopolists \citep[see, e.g.,][]{rochet87,rochet1998ironing,manelli07,daskalakis-etal2017,haghpanah21,lahr24}. \citet{rochet1998ironing} characterize optimality using mean-preserving spreads within bunching regions, and \citet{daskalakis-etal2017} obtain stochastic-dominance conditions for the optimality of finite-menu mechanisms. We develop related conditions for delegation without transfers.

\section{Model}\label{sec:model}

A principal chooses an action $a\in \R^n$. An agent is privately informed about the state of the world $s\in S$, where $S\subseteq \R^n$ is compact and convex and has non-empty interior. The agent's payoff $u_A(a,s)$ and the principal's payoff $u_P(a,s)$ depend on both the action and the state of the world. Specifically,
\begin{align}
u_A(a,s)&:= a\cdot s + b(a),\label{eq:agent_pref}
\end{align}
 where $b:\R^n\rightarrow \R$ is strictly concave and differentiable, and we assume $\max_{a\in\R^n} a\cdot s+b(a)$ has a solution for all $s\in\R^n$.

\begin{example}
A common specification assumes that the agent's preferences are Euclidean, so that the payoff from action $a$ in state $s$ is $-\norm{a-s}^2 = 2a\cdot s - \norm{a}^2 - \norm{s}^2$. After dividing by two and adding a term that does not depend on the action, and hence does not change preferences over actions, this is an instance of our payoffs.
\end{example}

 We assume that the principal's payoff $u_P(a,s)$ is continuous in $a$ and $s$,  concave and differentiable in $a$ for each $s$, and the gradient $\nabla_a u_P(a,s)$ is differentiable in $s$ for all $a$. In \autoref{sec:optimal}, we will use the additional structure that we obtain by imposing
 \begin{assumption}\label{ass:principal_pref}
 The principal's payoff is
 \begin{align}
   u_P(a,s)&:=a\cdot g(s) + b(a),\label{eq:principal_pref}
 \end{align}
 where $g: S\rightarrow \R^n$ is continuously differentiable.
 \end{assumption}

We assume that the state $s$ is distributed according to a probability distribution $F$ with continuously differentiable density $f$ and support $S$.
The principal aims to maximize her expected payoff and can design arbitrary mechanisms.

The revelation principle applies and we define a \emph{mechanism} to be a function $m:S\rightarrow \Delta(\R^n)$ such that all expected payoffs are finite and integrable.\footnote{We denote by $\Delta(\R^n)$ the set of Borel probability measures on $\R^n$.} To simplify notation, we extend the domain of $b(\cdot)$ and $u_i(\cdot,s)$ linearly to include probability distributions over $\R^n$, so that $b(m(s))=\E_{m(s)}[b(a)]$ and analogously for $u_i(\cdot,s)$. A mechanism is \emph{incentive compatible} if for all $s$ and $s'$ in $S$,
\begin{align*}
    u_A(m(s),s)\ge u_A(m(s'),s).
\end{align*}

\paragraph{Discussion} 

To formulate a tractable model of multidimensional delegation, we opted for a formulation that yields a simple characterization of incentive compatible mechanisms instead of going for the most general formulation.

For example, parts of the literature on delegation impose exogenous constraints on the set of feasible actions. While such constraints can be relevant for some applications, we opted to not impose them to obtain a simpler characterization of feasible mechanisms. The general approach of characterizing feasible mechanisms via indirect utilities is still useful in the presence of such constraints, but the indirect utilities must then satisfy additional constraints.

Also, we allow for stochastic mechanisms even though from a practical perspective we might sometimes mainly be interested in deterministic mechanisms. Knowing that specific deterministic mechanisms are optimal in the larger class of stochastic mechanisms is interesting from a theoretical perspective and can be useful for applications of the framework to dynamic extensions \citep[e.g.,][]{ali23}. It also yields a clean characterization of incentive compatible mechanisms because it makes the set of feasible mechanisms convex. The lotteries used by stochastic mechanisms in our framework play a role similar to money burning that is allowed in parts of the literature \citep[e.g.,][]{AB:13,AE:17} and we would obtain the same characterization of feasible mechanisms if we instead allowed for money burning.

Finally, we assume that the agent's preference can be captured by a payoff that is affine in the state for each action. This assumption is common in delegation and mechanism design more generally. There is a sense in which the problem becomes a lot more difficult if preferences are not affine in the state: \citet{kartik24} show that, when the agent chooses among lotteries over a multidimensional action space, this affine form is in a sense essential for local incentive constraints being sufficient for global incentive compatibility. More abstract action spaces can in principle be handled in a similar way, by relabeling each action according to the payoff it induces for the agent. This again yields a payoff of the form we assume, but additional complications arise: the relabeled actions need not fill out the whole space, and distinct actions may become indistinguishable for the agent while still mattering to the principal.

\section{Characterizing incentive-compatible mechanisms}\label{sec:ic}

We follow \citet{rochet87} and characterize the set of incentive-compatible mechanisms in terms of their indirect utilities.
To any incentive-compatible mechanism $m$ corresponds an \emph{indirect utility} $U:\R^n\rightarrow\R$ defined by 
\[U(s):=\max_{s'\in S} \E[m(s')]\cdot s + b(m(s')).\]
Which indirect utilities correspond to some incentive-compatible mechanism? First, any indirect utility is convex as the maximum of a family of functions that are affine in the state $s$. In some mechanism design settings with transfers, this is all it takes to be the indirect utility of some incentive-compatible mechanism \citep[e.g.,][]{rochet1998ironing}. However, in the absence of transfers additional restrictions apply. In particular, if the principal cannot transfer money to the agent, the agent's utility cannot be higher than if he was free to choose his action. Defining the \textit{full discretion payoff} $h:\R^n\rightarrow \R$ (i.e., the maximum payoff a type could ever get) by
 \[h(s):=\max_{a\in \R^n} a\cdot s + b(a),\]
$U\le h$ is therefore clearly necessary.\footnote{We denote the pointwise order by $\le$, so $U\le h$ means $U(s)\le h(s)$ for all $s\in\R^n$.} The following result shows that these two conditions characterize the set of feasible indirect utilities.

\begin{lemma}\label{lemma:feasible_set}
An indirect utility $U$ corresponds to an incentive-compatible mechanism if and only if $U$ is convex and lies below the full discretion payoff: $U\le h$.
\end{lemma}

\begin{figure}[t]
\centering
\begin{tikzpicture}[x=0.95cm, y=0.85cm]

\draw[guide] (-2,-1.95) -- (-2,0.85);
\draw[guide] (3,-1.95)  -- (3,1.475);

\begin{scope}
\clip (-4.3,-2.0) rectangle (3.35,2.6);
  \fill[viol!18, draw=viol!45, line width=0.3pt]
      (-5.8439,8.5378) -- (-2.1561,1.1622)
      -- plot[domain=-2.1561:-5.8439, samples=80] (\x,{0.25*\x*\x}) -- cycle;
  \draw[hcurve] plot[domain=-6.1:3.35, samples=160] (\x,{0.25*\x*\x});
  \draw[uext]   (-2,0.85) -- (-6.1,9.05);
  \draw[ucurve] plot[domain=-2:3, samples=120] (\x,{0.425*\x*\x-0.3*\x-1.45});
  \fill[ucol]   (-2,0.85) circle (1.6pt);
\end{scope}

\draw[black!45, line width=0.5pt] (-4.3,-1.95) -- (3.35,-1.95);
\draw[line width=1.7pt] (-2,-1.95) -- (3,-1.95);
\draw[black!45] (-2,-1.83) -- (-2,-2.07);
\draw[black!45] (3,-1.83)  -- (3,-2.07);
\node[ann, below] at (-2,-2.05)  {$\underline s$};
\node[ann, below] at (3,-2.05)   {$\overline s$};
\node[ann, below] at (0.5,-2.05) {$S$};

\node[hcol, ann, right, inner sep=2.5pt] at (3.08,2.25) {$h$};
\node[ucol, ann, right,      inner sep=2.5pt] at (3,1.475)    {$U$};
\end{tikzpicture}
\caption{The function $U$ satisfies $U(s)\le h(s)$ for all $s\in S$ but does not correspond to a feasible mechanism because every convex extension of $U$ lies above the dashed ray.}
\label{fig:domain_indirect_util}
\end{figure}

Intuitively, if $U$ is convex then it would correspond to an incentive-compatible mechanism if transfers were available and the agent had quasi-linear preferences. If the required transfers are all negative (that is, payments from the agent to the principal) then we can use the agent's risk aversion (coming from the strict concavity of $b$) to simulate these transfers via stochastic actions. If $U\le h$, one can indeed show that the required transfers are negative. 

This last step relies on our choice that the domain of $U$ and $h$ is large enough. For example, \citet{rochet1998ironing} define $U$ only as a function on $S$; however,  requiring only that $U(s)\le h(s)$ for all $s\in S$ is not sufficient for our conclusion. \autoref{fig:domain_indirect_util} illustrates a convex function $U$ defined on $S$ which lies below $h$ on all of $S$, but which does not correspond to a mechanism because it cannot be extended to a convex function that lies below $h$ on all of $\R^n$ (which must be satisfied for any indirect utility by the simple implication in \autoref{lemma:feasible_set}). 

\begin{proof}
We have argued for the necessity of these conditions before the statement of the lemma. For the converse direction, let us first recall basic observations from convex analysis. The convex conjugate of a function $U$ is denoted by $U^*$ and defined by $U^*(a):=\sup_{s\in\R^n} a\cdot s - U(s)$. We will use the following facts, which follow immediately from this definition: (i) $h=(-b)^*$, (ii) $U\le h$ implies $h^*\le U^*$, and (iii) $a\in\partial U(s)$ implies $U^*(a)=a\cdot s -U(s)$.\footnote{Here, $\partial U(s)$ denotes the subdifferential of $U$ at $s$. To see (iii), note that the definition of $U^*$ implies $U^*(a)\ge a\cdot s -U(s)$. Conversely, $a\in\partial U(s)$ implies that for all $s'$, 
$a\cdot s-U(s)\ge a\cdot s'-U(s')$. Taking the supremum of the right-hand side with respect to $s'$ yields $a\cdot s-U(s) \ge U^*(a)$.}

Suppose $U$ is convex and satisfies $U\le h$. Fix $s\in S$ and $a\in\partial U(s)$. Facts (iii), (ii) and (i), together with $(-b)^{**}=-b$, give
\[ U(s)=a\cdot s-U^*(a)\le a\cdot s-h^*(a)=a\cdot s+b(a), \]
which is the payoff type $s$ obtains from action $a$ for sure. Because $b$ is strictly concave, we can therefore assign to type $s$ a lottery $m(s)$ with expected value $a$ that yields payoff $u_A(m(s),s)=U(s)$.
It remains to show that $m$ is incentive compatible. For all $s$ and $s'$,
\begin{align*}
    &u_A(m(s),s) = U(s) \ge U(s') + \E[m(s')]\cdot (s-s') \\
    = &\E[m(s')]\cdot s' + b(m(s')) + \E[m(s')]\cdot (s-s') = u_A(m(s'),s),
\end{align*}
where the first inequality follows since $\E[m(s')]\in \partial U(s')$.
\end{proof}

\autoref{fig:example_feasible_mechanisms} illustrates \autoref{lemma:feasible_set} for one-dimensional types and quadratic payoffs. It shows two indirect utilities that correspond to incentive-compatible mechanisms. In \autoref{fig:first}, all types between $s_1$ and $s_2$ obtain their full discretion payoff and $U$ is affine below $s_1$ and above $s_2$. This indirect utility can be obtained by letting types choose their preferred action from the interval of deterministic actions $[s_1,s_2]$. The indirect utility in \autoref{fig:third} contains an affine piece that lies strictly below the graph of $h$. This part of the indirect utility corresponds to types that obtain a stochastic action, which yields no type its full discretion payoff. 

\begin{figure}
\centering
\begin{subfigure}[t]{0.45\textwidth}\centering
\begin{tikzpicture}[x=1.25cm, y=1.1cm]
\begin{scope}
\clip (-2.4,-0.32) rectangle (2.4,1.8);
  \draw[hcurve] plot[domain=-2.4:2.4, samples=120] (\x,{0.5*\x*\x});
  \draw[ucurve] (-2.4,1.9) -- (-1,0.5)
                plot[domain=-1:1, samples=60] (\x,{0.5*\x*\x});
  \draw[ucurve] (1,0.5) -- (2.4,1.9);
\end{scope}
\draw[saxis] (-2.4,-0.28) -- (2.4,-0.28);
\draw[sint]  (-1,-0.28) -- (1,-0.28);
\draw[saxis] (-1,-0.20) -- (-1,-0.36);
\draw[saxis] (1,-0.20)  -- (1,-0.36);
\node[ann, below] at (-1,-0.36) {$s_1$};
\node[ann, below] at (1,-0.36)  {$s_2$};
\node[hcol, ann] at (-1.55,1.58) {$h$};
\node[ucol, ann] at (2.05,1.18)  {$U$};
\end{tikzpicture}
\caption{Interval delegation}
\label{fig:first}
\end{subfigure}
\hfill
\begin{subfigure}[t]{0.45\textwidth}\centering
\begin{tikzpicture}[x=1.25cm, y=1.25cm]
\begin{scope}
\clip (-2.4,-0.32) rectangle (2.4,1.8);
  \draw[hcurve] plot[domain=-2.4:2.4, samples=120] (\x,{0.5*\x*\x});
  \draw[ucurve] plot[domain=-2.4:0.3, samples=70] (\x,{0.5*\x*\x});
  \draw[ucurve] (0.3,0.045) -- (0.8786,0.2186) -- (1.2368,0.5768) -- (2.4,2.729);
  \draw[kink] (0.8786,0.2186) circle (2.1pt);
  \draw[kink] (1.2368,0.5768) circle (2.1pt);
\end{scope}
\draw[saxis] (-2.4,-0.28) -- (2.4,-0.28);
\end{tikzpicture}
\caption{A stochastic mechanism}
\label{fig:third}
\end{subfigure}

\caption{Examples of indirect utilities. The blue curves show the function $h$ for one-dimensional types and quadratic payoffs. The green curves show indirect utilities corresponding to incentive-compatible mechanisms.} 
\label{fig:example_feasible_mechanisms}
\end{figure}

\autoref{lemma:feasible_set} also illuminates a connection between optimal delegation problems and certain Bayesian persuasion problems established  in \citet{KZ:19} and \citet{kleiner21}. In a mean-measurable persuasion problem \citep[see][]{gentzkow2016rothschild,Kolo:18}, the designer is choosing a Blackwell experiment that reveals information about a one-dimensional state to a receiver. The designer's payoff is determined by the expected value of the receiver's posterior belief; the designer's problem is therefore to choose a distribution of posterior means to maximize her expected payoff. 

A distribution of posterior means, $G$, can be induced by some experiment if and only if it is a mean-preserving contraction of the prior $F$ \citep{Kolo:18}; if the support of $F$ is contained in $[\underline{s},\overline{s}]$ this in turn is equivalent to 
\[U_G(s):=\int_{\underline{s}}^s G(x)\,\mathrm dx\le \int_{\underline{s}}^s F(x)\,\mathrm dx=:U_F(s)\]
for all $s$ and with equality for $s=\overline{s}$. Clearly, $U_G$ and $U_F$ are increasing convex functions (satisfying $U_G(\underline{s})=U_F(\underline{s})$, $U_G'(\underline{s})=U_F'(\underline{s})$, $U_G(\overline{s})=U_F(\overline{s})$, and $U_G'(\overline{s})=U_F'(\overline{s})$). Since any choice of $U_G$ satisfying these properties determines a feasible choice of $G$, one can model the designer in the persuasion problem as choosing a convex function that lies below a given convex function (and coincides with this function at $\underline{s}$ and $\overline{s}$), see \citet{gentzkow2016rothschild}. \autoref{lemma:feasible_set} illustrates therefore the close connection between the persuasion problem and the one-dimensional delegation problem. \citet{KZ:19} and \citet{kleiner21} show how the additional requirement on $U_G$ can be satisfied without loss of optimality (by choosing the interval $[\underline{s},\overline{s}]$ a lot larger than the support of $F$), in which case the two problems become isomorphic. 
Turning to multidimensional problems, \autoref{lemma:feasible_set} shows that multidimensional delegation problems can also be formulated as the designer choosing a convex function below a given function. However, multidimensional versions of the persuasion problem do not correspond to choosing convex functions in general \citep{kleiner24} and, therefore, the two problems are not equivalent in higher dimensions.

\section{Characterizing optimal mechanisms}\label{sec:optimal}

We characterize optimal mechanisms in this section under the maintained assumption that the principal's preferences satisfy \autoref{ass:principal_pref}. To do so, we first formulate the principal's problem in terms of indirect utilities. We then state the main characterization of optimal mechanisms in \autoref{sec:optimal_mechanisms} and illustrate the result for particular mechanisms. Finally, we outline the proof of the main result in \autoref{sec:proof_sketch}.

\subsection{Formulating the principal's problem}
\label{sec:formulating_problem}

Consider an indirect utility  $U$ that corresponds to some incentive-compatible mechanism. In general, there are many incentive-compatible mechanisms that induce the same indirect utility; however, under our assumptions on payoffs, all such mechanisms induce the same payoff for the principal. 
\begin{example}
If the agent's payoff is quadratic (i.e., $u_A(a,s)=a\cdot s -\frac{1}{2}\norm{a}^2$), an indirect utility determines, almost everywhere, the mean action and the total variance of each assigned lottery. Under \autoref{ass:principal_pref}, these two moments also determine the principal's payoff, so she is indifferent among all mechanisms that induce the same indirect utility. If instead she weighs deviations differently across dimensions, the direction of randomization matters; \autoref{sec:general_payoffs} studies that case.
\end{example}
To show this more generally, let $m$ be an incentive-compatible mechanism with corresponding indirect utility  $U$. Using $\nabla U(s)=\E[m(s)]$ (by the envelope theorem) and $U(s) = \nabla U(s)\cdot s + b(m(s))$, the principal's payoff from mechanism $m$ in state $s$ is completely determined by $U$:
\begin{align}\label{eq:reformulate_U}
    \E[m(s)]\cdot g(s) + b(m(s)) = \nabla U(s) \cdot [g(s)- s]+ U(s). 
\end{align}
This observation implies in addition that the principal's payoff is a linear function of $U$. Therefore, a solution to the principal's problem can be found at an extreme point of the feasible set. Since in higher dimensions the set of extremal convex functions is a lot richer than in one-dimensional settings \citep[see][]{johansen1974extremal,bronshtein1978extremal,lahr24}, we proceed to characterize optimal mechanisms.

One disadvantage of the reformulated objective in \eqref{eq:reformulate_U} is that the indirect utility $U$ and its gradient $\nabla U$ appear. As is standard in multidimensional mechanism design \parencite[see, for example,][]{rochet1998ironing}, we can use the divergence theorem to reformulate the objective function so that only $U$ appears.
For intuition, suppose first $S=[\underline{s},\overline{s}]$. 
We can use integration by parts to obtain 
\begin{align*}%
&\int U(s)f(s)+ U'(s) [g(s)- s] f(s) \,\mathrm ds \nonumber\\
 = &\int U(s) \big[f(s)-[[g(s)- s] f(s)]'\big] \,\mathrm ds+ U(\overline{s})(g(\overline{s})-\overline{s})f(\overline{s}) - U(\underline{s})(g(\underline{s})-\underline{s})f(\underline{s}) .
\end{align*}
In higher dimensions, we can use the divergence theorem (or integration by parts dimension-by-dimension) to obtain the following analogue:
\begin{align*}
    &\int \big[ U(s)+ \nabla U(s)\cdot [g(s)- s] \big] \,\mathrm dF(s)\\
    =& \int U(s) \big[  f(s) - \diver[(g(s)- s)f(s)] \big] \,\mathrm ds + \int_{\bd S} U(s)[g(s)- s]f(s) \cdot \n_S(s) \,\mathrm d\H(s),
\end{align*}
where $\diver$ denotes the divergence of a function, for any set $A$ its boundary is denoted by $\bd A$, $\H$ denotes the $n-1$-dimensional Hausdorff measure on the boundary of $S$, and $\n_S(s)$ denotes the unit outward normal vector to the convex set $S$ at $s\in \bd S$.
Note that the second term vanishes if the density $f$ is zero on the boundary of the type space, in which case the following expressions simplify. 

Collecting the two weights that multiply $U(s)$ in this identity, define
\begin{align} \label{eq:def_nu}
\nu(s):= \begin{cases}
f(s) - \diver[(g(s)- s)f(s)]     &\text{ if } s\in\interior S \\
[g(s)- s]f(s) \cdot \n_S(s) &\text{ if } s\in\bd S,
\end{cases}
\end{align}
and let $\lambda$ be the Lebesgue measure on $S$ plus the Hausdorff measure on the boundary of $S$,
 so that the principal's payoff from an indirect utility $U$ is $\int U(s)\nu(s) \,\mathrm d\lambda(s)$. Denoting by $\mu$ the (signed) measure with density $\nu$ relative to $\lambda$, that is, $\mu(E) = \int_E \nu(s) \,\mathrm d\lambda(s)$, we can write the principal's problem as
\begin{align*}\label{eq:primal}
    &\max_{U \text{ convex}} \int U(s) \,\mathrm d\mu(s) \tag{P}\\*
    & \text{s.\,t. } U\le h.
\end{align*}

\paragraph{Geometric intuition}
Since $\int U \,\mathrm d\mu = \int U(s)\,\nu(s)\,\mathrm d\lambda(s)$, the objective in \eqref{eq:primal} weights the utility of each type $s$ by $\nu(s)$, and we can interpret $\nu$ as a \emph{virtual value for utility}. Heuristically, $\nu(s)$ measures how much the principal's payoff increases if the indirect utility of type $s$ is increased, but where types on the boundary get extra weight.
Our formulation of the principal's problem suggests a geometric approach to think about optimal mechanisms. According to \eqref{eq:primal}, the principal aims to give those types $s$ with $\nu(s)>0$ a utility as high as possible, where the full discretion payoff $h$ provides an upper bound. Conversely, for types with $\nu(s)<0$ the principal aims to assign utilities as low as possible, but convexity of $U$ provides a bound on how low the utility of a given type can be given other types' utilities.

Let us illustrate with a one-dimensional example, where we assume that $\nu$ is  positive on some interval of intermediate types and negative outside this interval (see \autoref{fig:nu_pos_on_interval}). As we discuss below, assumptions commonly used in the literature are sufficient for this. The principal will want to give low and high types---for whom $\nu(s)<0$---utilities as low as possible until the convexity constraint binds, which makes $U$ affine below some cutoff and above some cutoff. For the types in between, $\nu(s)>0$, and the principal sets their utility as high as possible, hence equal to their full discretion payoff. An optimal mechanism can therefore be chosen to be interval delegation, as illustrated in \autoref{fig:first}. We establish this formally in \autoref{cor:suff_one_dim}.

\begin{figure}
\centering
\begin{tikzpicture}[x=1.30cm, y=0.5cm]
\begin{scope}
\clip (-3.2,4.5) rectangle (3.45,8.2);
  \draw[hcurve] plot[domain=-3:3, samples=220]
     (\x,{6.4 + 3*(1.05*exp(-0.45*\x*\x) + 0.34*exp(-9*(\x-1.3)*(\x-1.3))
                   + 0.16*exp(-8*(\x+1.1)*(\x+1.1)) - 0.52)});
\end{scope}
\draw[saxis] (-3.15,6.4) -- (3.3,6.4);
\node[hcol, ann] at (3.28,5.05) {$\nu$};
\begin{scope}
\clip (-3.2,-0.45) rectangle (3.45,4.0);
  \draw[hcurve] plot[domain=-3:3, samples=160] (\x,{0.5*\x*\x});
  \draw[ucurve] (-3,2.5) -- (-1,0.5)
                plot[domain=-1:1.2, samples=60] (\x,{0.5*\x*\x});
  \draw[ucurve] (1.2,0.72) -- (3,2.88);
\end{scope}
\draw[saxis] (-3.15,-0.3) -- (3.3,-0.3);
\node[hcol, ann] at (2.35,3.55)  {$h$};
\node[ucol, ann] at (-2.50,1.35) {$U$};
\node[ann] at (3.3,-0.66) {$s$};
\end{tikzpicture}
\caption{A one-dimensional example illustrating why interval delegation is optimal if $\nu$ is positive on some interval and negative otherwise. }
\label{fig:nu_pos_on_interval}
\end{figure}

\begin{example}
\label{ex:uniform}
We return to the example from the introduction. Let $S=[0,1]^2$, let $F$ be uniform on $S$, and let preferences be quadratic with a common bias in the two dimensions:
\[ b(a)=-\frac12\norm{a}^2, \qquad g(s)=s-(\beta, \beta),\qquad \beta>0. \]
The agent's ideal action is $s$, whereas the principal's is $s-(\beta,\beta)$. The full-discretion payoff is $h(s)=\norm{s}^2/2$, and the virtual value for utility in \eqref{eq:def_nu} simplifies to
\begin{align}\label{eq:nu_uniform}
\nu(s)=
\begin{cases}
1 &\text{ if } s\in\interior S,\\
\beta &\text{ if } s \text{ lies on the west or south edge},\\
-\beta &\text{ if } s \text{ lies on the north or east edge}.
\end{cases}
\end{align}
The principal therefore wants to lower utility on the north and east edges and raise it elsewhere. This suggests pooling types near those edges with types in the interior.
\end{example}

\subsection{Optimal mechanisms}
\label{sec:optimal_mechanisms}

We next characterize when a given mechanism is optimal. To do so, we introduce additional notation: Given an indirect utility $U$, we let $\Q$ denote the partition of $\R^n$ into maximal sets on which $U$ is affine. Two types belong to the same partition element exactly when they can be assigned the same lottery over actions, so the partition elements are the pooling regions of the mechanism. This description is unambiguous at every type at which $U$ is differentiable. At a kink, $U$ is affine on several maximal sets containing that type, and we specify in \autoref{sec:partition} which of them the type is assigned to.\footnote{The set of kinks is Lebesgue-null since $U$ is convex, but it can still carry positive mass under $\mu$ because of the extra measure on the boundary of $S$. If the density $f$ vanishes on $\bd S$, the issue does not arise: $\mu$ is then absolutely continuous, $U$ is differentiable at $|\mu|$-almost every type, and $\Q$ is the partition into the level sets of $\nabla U$.}

We denote by $\{\mu|_Q\}_{Q\in \Q}$ a conditional measure of $\mu$ given $Q$. Given $\Q$, we say that some property holds for $|\mu|$-a.e.\ $Q\in\Q$ if it holds for every $Q$ outside a subcollection of $\Q$ whose union has $|\mu|$-measure zero.

For any $Q\in \Q$, there is at most one $s\in Q$ such that $U(s)=h(s)$: $U$ is affine on a convex set containing $Q$, $U\le h$, and $h$ is strictly convex.\footnote{Since $h=(-b)^*$ and conjugacy exchanges differentiability and strict convexity, $h$ is differentiable and strictly convex; see \citet[E.4.1]{hiriart2004}.} If such a point exists, we let $\delta_Q$ be the point mass of mass $\mu|_Q(Q)$ at it; if $U(s)<h(s)$ for all $s\in Q$, then $\delta_Q$ is the zero measure. For any two signed measures $\alpha$ and $\beta$, we write $\alpha\le_{cx}\beta$ and say that $\alpha$ is smaller than $\beta$ in the \emph{convex order} if $\int c(s)\,\mathrm d\alpha\le\int c(s)\,\mathrm d\beta$ for every convex function $c:\R^n\rightarrow\R$. We can now state the main characterization result:

\begin{theorem}\label{thm:main_result}
Suppose the principal's preferences satisfy \autoref{ass:principal_pref} and let $U$ be a feasible indirect utility. Then $U$ is optimal if and only if for $|\mu|$-a.e. $Q\in \Q$, $\mu|_Q(Q)\ge 0$ and $\mu|_Q\le_{cx} \delta_Q$.
\end{theorem}

Why are the conditions in \autoref{thm:main_result} sufficient for $U$ to be optimal? Consider a partition element $Q\in \Q$ and suppose $\delta_Q$ is a positive point mass at $s^*$. Then any feasible indirect utility $V$ will be convex, lie below $U$ at $s^*$ with $U(s^*)=h(s^*)$, and $U$ will be an affine function on $Q$. Together with the fact $\mu|_Q \le_{cx} \delta_Q$ this implies
\[ \int V \,\mathrm d\mu|_Q \le \int V \,\mathrm d\delta_Q \le \int U \,\mathrm d\delta_Q  = \int U \,\mathrm d\mu|_Q. \]
Indeed, the first inequality holds because $V$ is convex and $\mu|_Q \le_{cx} \delta_Q$, the second because $V(s^*)\le U(s^*)$ and $\delta_Q$ is a positive point mass at $s^*$, and the equality because $U$ is affine on $Q$ %
and $\mu|_Q\le_{cx} \delta_Q$.
It follows that, restricted to $Q$, the principal's objective is at least as large under $U$ as under any feasible $V$. If this holds for almost all $Q$, then $U$ is optimal.
The conclusion that these conditions are also necessary shows that the problem can in some sense be decomposed: whenever the principal can improve $U$ conditional on $Q$, she can extend this improved version to a feasible indirect utility that yields unconditionally a higher payoff.

\paragraph{Pooling and the sign of virtual values.}
The characterization also helps us find candidate optimal mechanisms. Pooling must combine types with positive and negative virtual values: the principal sacrifices utility for the former in order to reduce utility for the latter, and this trade-off must be justified within each pooling region, not just across the mechanism as a whole. In particular, when $\nu$ is positive in the interior and negative only on parts of the boundary, every nontrivial pool containing interior types must reach a negatively weighted part of the boundary. This suggests starting from those boundary types and constructing pooling regions that extend into the interior. The sign pattern thus narrows the search for an optimum before its allocation has been determined, and we use this insight below.

\subsection{Optimal delegation sets}

A particularly simple mechanism has the principal delegate the decision to the agent, possibly restricting the agent's action to some subset of actions. Note that any deterministic mechanism can be implemented as an indirect mechanism in this way.
For a nonempty closed set $A\subseteq\R^n$, \emph{delegating to $A$} means that the agent chooses his preferred action among the ideal actions of types in $A$; the types indexing the permitted actions need not belong to $S$. We say that delegating to $A$ is optimal if this mechanism is optimal. In the following we specialize \autoref{thm:main_result} to these mechanisms.

On the real line, convex order can be checked by comparing total masses, first moments, and the integrals of the convex functions $\max\{s-t,0\}$ for all $t$ \parencite[][Theorem 3.A.1]{shaked2007stochastic}. This simplifies \autoref{thm:main_result} whenever $U$ is affine on at most one-dimensional sets; as we will see, this is useful even if the type space is multidimensional. We first consider interval delegation with one-dimensional types \parencite[for earlier characterizations, see][]{AM:08,AB:13}. 

\begin{corollary}\label{cor:onedimension}
Suppose $n=1$ and $s_1,s_2\in S$ with $s_1<s_2$. Delegating to $[s_1,s_2]$ is optimal if and only if
\begin{enumerate}[label=(\roman*)]
\item $\nu(s)\ge 0$ for all $s\in[s_1,s_2]$,
\item $\int_{[s,\overline{s}]}(x-s) \,\mathrm d\mu(x)\le 0$ for all $s\in[s_2,\overline{s}]$ with equality for $s=s_2$, and
\item $\int_{[\underline{s},s]}(s-x) \,\mathrm d\mu(x)\le 0$ for all $s\in[\underline{s},s_1]$ with equality for $s=s_1$.
\end{enumerate}
\end{corollary}

The partition induced by this mechanism consists of the two tails $(-\infty,s_1]$ and $[s_2,\infty)$ and the singletons between them. Condition (i) gives the comparison on the singletons, and conditions (ii) and (iii) express the convex-order comparisons on the two pooling regions.

\begin{figure}
\centering
\begin{tikzpicture}[x=1.25cm, y=0.65cm]
\begin{scope}
\clip (-2.05,-0.62) rectangle (3.4,4.8);
  \draw[hcurve] plot[domain=-2:3.3, samples=160] (\x,{0.5*\x*\x});
  \draw[ucurve] (-2,1.5) -- (-1,0.5)
                plot[domain=-1:1, samples=50] (\x,{0.5*\x*\x});
  \draw[ucurve] (1,0.5) -- (3.3,2.8);
  \draw[utilta] (2,1.5) -- (3,3.1);
  \draw[utiltb] (2,1.5) -- (3,4.5);
  \draw[->, black!60, line width=0.6pt] (2.62,2.12) to[bend right=28] (2.72,3.50);
  \fill[ucol] (2,1.5) circle (1.7pt);
\end{scope}
\draw[saxis] (-2.05,-0.55) -- (3.4,-0.55);
\draw[sint]  (-1,-0.55) -- (1,-0.55);
\foreach \s in {-1,1,2,3} \draw[saxis] (\s,-0.42) -- (\s,-0.68);
\node[ann, below] at (-1,-0.68) {$s_1$};
\node[ann, below] at (1,-0.68)  {$s_2$};
\node[ann, below] at (2,-0.68)  {$s$};
\node[ann, below] at (3,-0.68)  {$\overline s$};
\end{tikzpicture}
\caption{Optimality of interval delegation}
\label{fig:economic_iterpretation_interval}
\end{figure}

\autoref{fig:economic_iterpretation_interval} illustrates condition (ii) of \autoref{cor:onedimension}. For $s>s_2$, a small upward tilt of the indirect utility starting at $s$ raises the utility of every type $x\ge s$ in proportion to $x-s$. Its effect on the principal's expected payoff is therefore proportional to $\int_{[s,\overline{s}]}(x-s) \,\mathrm d\mu(x)$, and condition (ii) rules out a profitable tilt. Equality at $s_2$ expresses in addition optimality with respect to tilting the utility downwards at $s_2$, which is feasible since $U$ is strictly convex in a neighborhood to the left of $s_2$.

Interestingly, the conditions identified in \autoref{cor:onedimension} are in our setting equivalent to the ones obtained in \citet[Proposition 2a]{AB:13}. At first glance, this might be surprising since we characterize optimality of interval delegation in the class of stochastic mechanisms and \citeauthor{AB:13} characterize optimality in the class of deterministic mechanisms (and stochastic mechanisms can do strictly better in general). \autoref{fig:economic_iterpretation_interval} illustrates why the conditions are the same: Suppose the principal strictly benefits from deviating to the dashed indirect utility, which represents a stochastic mechanism. Since her payoff is linear in $U$, the arguments in the previous paragraph imply that she also benefits from deviating to the dotted indirect utility. Since the dotted linear utility corresponds to a deterministic mechanism, we conclude that conditions (ii) and (iii) in \autoref{cor:onedimension} are necessary for interval delegation to be optimal in the class of deterministic mechanisms (and necessity of condition (i) can be shown easily).

For applications, it can be useful to have simple sufficient conditions that ensure interval delegation is optimal for some interval. The following corollary provides such a sufficient condition.

\begin{corollary}\label{cor:suff_one_dim}
Suppose $n=1$ and there is an interval $I\subseteq S$ such that the restriction of $\mu$ to $I$ is nonnegative and its restriction to $S\setminus I$ is nonpositive. Then delegating to a closed interval, possibly a singleton, is optimal.
\end{corollary}

The proof starts from an optimal indirect utility and uses \autoref{thm:main_result} to show that at most two of its pools carry negative virtual weight. Alternatively, one can prove it using the geometric ideas discussed in \autoref{sec:formulating_problem}.

In the present payoff setting, the hypothesis of \autoref{cor:suff_one_dim} includes the global sign condition in \citet[Proposition 2(a)]{ABF:18}, which requires $\nu$ to be nonnegative throughout the interior of $S$. A simple implication of our result is the following, which can be useful for applications \citep[see][for a related condition for settings with quadratic preferences]{KM:09}.

\begin{corollary}\label{cor:log_concave}
Suppose $n=1$ and $g(s)=s+\beta$ for some $\beta\in\R$. If $f$ is log-concave, then delegating to a closed interval, possibly a singleton, is optimal.
\end{corollary}

For a one-dimensional type space, the approach used in \autoref{cor:onedimension} simplifies the conditions in \autoref{thm:main_result} for any mechanism, not just interval delegation. The same idea is useful with multidimensional types. Suppose $b(a)=-\norm{a}^2/2$, so each type's ideal action equals his type, and let $A$ be nonempty, closed and convex. Delegating to $A$ assigns each type his Euclidean projection onto $A$. For each $a\in\bd A$, let $Q_a:=a+N_A(a)$ denote the set of types whose projection is $a$, where $N_A(a)$ is the normal cone to $A$ at $a$. If the boundary of $A$ is differentiable, $Q_a$ is an outward normal ray. For any $s\in Q_a$, the translated ray $s+N_A(a)$ consists of the types on this pooling ray that lie beyond $s$, moving away from $A$.

\begin{figure}
\centering
\includegraphics[trim={0 7bp 0 7bp},clip,width=.4\textwidth]{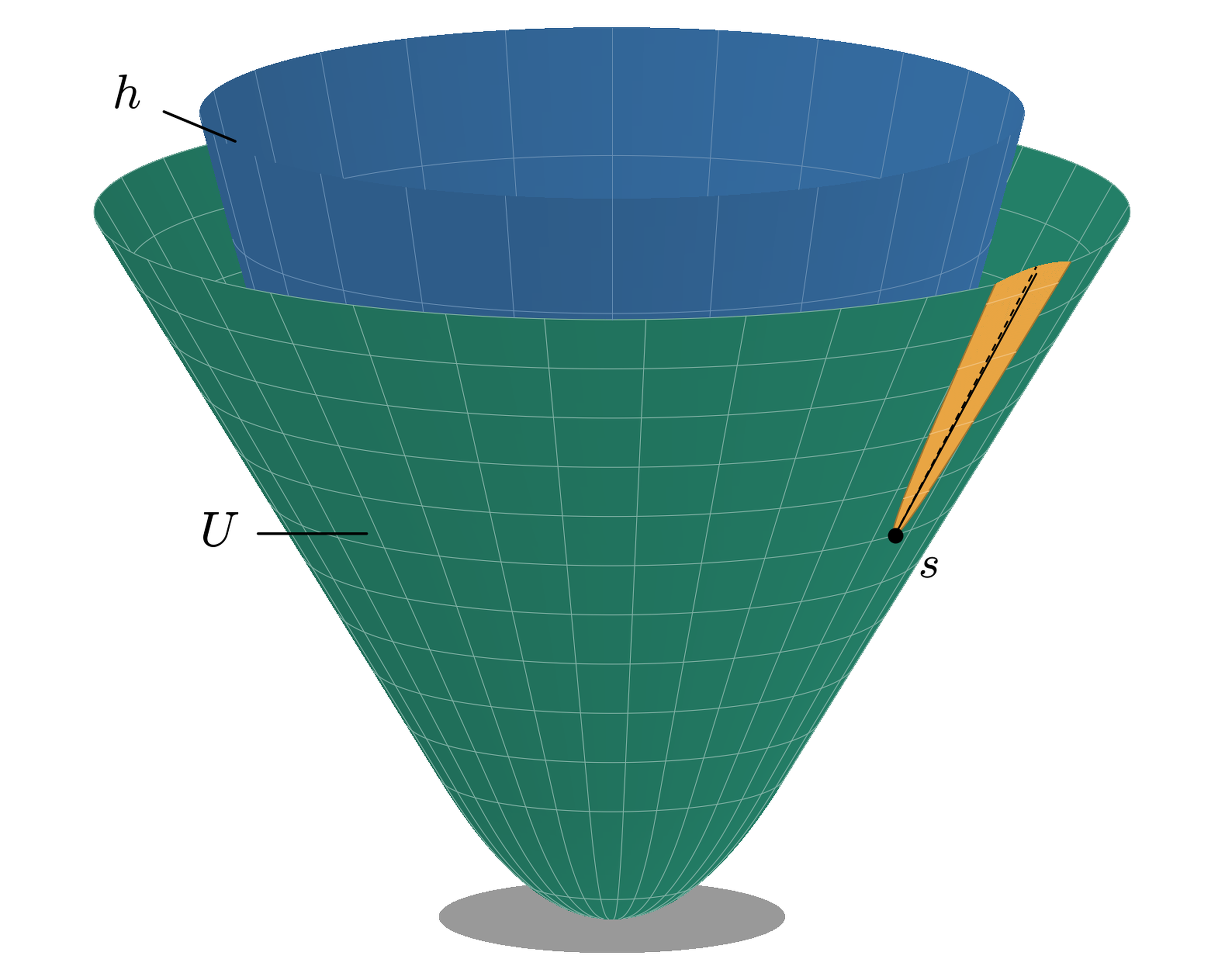}
\caption{Indirect utility of a convex delegation set and a stochastic perturbation. Adding a lottery raises utility to the orange planar patch; the dashed line shows the tilt along the ray.}
\label{fig:convex_delegation}
\end{figure}

\begin{corollary}\label{cor:convex_delegation}
Suppose $b(a)=-\norm{a}^2/2$ and $A\subseteq S$ is closed, convex, has nonempty interior and a differentiable boundary. Delegating to $A$ is optimal if and only if
\begin{enumerate}[label=(\roman*)]
\item $\nu(s)\ge 0$ for all $s\in\interior A$, and
\item for $|\mu|$-almost every normal ray $Q_a$, with $a\in\bd A$, we have $\mu|_{Q_a}(Q_a)\ge 0$ and
{\postdisplaypenalty=10000
\[ \int_{s+N_A(a)} \norm{x-s} \,\mathrm d\mu|_{Q_a}(x)\le 0 \qquad\text{for all } s\in Q_a, \]
with equality for $s=a$.}
\end{enumerate}
\end{corollary}

The conditions in \autoref{cor:convex_delegation} closely resemble those in \autoref{cor:onedimension}. Indeed, Condition (i) in either case requires that $\nu$ is nonnegative on the set of types that obtain their full discretion payoffs, and Condition (ii) (and Conditions (ii) and (iii), respectively) imposes that for each point on the boundary the analogous stochastic dominance condition holds.

The economic interpretation of condition (ii) is analogous to that of the interval delegation conditions. Along a pooling ray, an upward tilt starting at $s$ raises utility in proportion to the distance from $s$, and the integral in (ii) gives the corresponding conditional change in the principal's payoff. By adding a suitable lottery to the menu, the principal can produce such a tilt with an arbitrarily small utility increase confined to a small neighborhood of the black line segment in \autoref{fig:convex_delegation}. No deterministic mechanism can reproduce this change in indirect utility: an affine patch generated by a deterministic mechanism must extend to a plane tangent to $h$, whereas the orange patch extends to a plane strictly below $h$ everywhere. This illustrates the role of stochastic mechanisms in the multidimensional perturbation argument.

\paragraph{Budget delegation.}
We call a set of the form $\{a:w\cdot a\leq c\}$ a \emph{budget set}: it lets the agent choose any action whose weighted sum does not exceed $c$.
With quadratic payoffs and constant bias \(\bbeta\), an optimal budget must have weights positively proportional to \(\bbeta\). Other directions either restrict actions the wrong way or leave the agent unrestricted in a direction of disagreement, where an additional cap would improve the principal’s payoff. We therefore use weights \(\bbeta\): the budget restricts the component of the action along \(\bbeta\) and leaves the orthogonal components unrestricted. With equal biases, this is simply a cap on the sum of the actions.
The nontrivial pooling regions in a budget set are rays parallel to $w$, and, as for the normal rays in \autoref{cor:convex_delegation}, the comparisons required by \autoref{thm:main_result} are one-dimensional.

Write \(Z=\bbeta\cdot s\) and let \(Y\) be the component of the state orthogonal to \(\bbeta\). Types pooled together have the same \(Y\) and differ only in \(Z\). Optimality can therefore be characterized using the distribution of \(Z\) conditional on \(Y\).

\begin{corollary}\label{cor:flat_conditional}
Suppose $b(a)=-\norm{a}^2/2$ and $g(s)=s-\bbeta$ for some
$\bbeta\neq0$, and fix $c\in\R$. Let $Z=\bbeta\cdot s$ and
let $Y$ be the component of $s$ orthogonal to $\bbeta$.

Delegating to the budget set $\{a:\bbeta\cdot a\leq c\}$
is optimal if and only if, conditional on almost every $Y$,
the cap $c$ is optimal in the one-dimensional quadratic problem
with state $Z$ and principal's ideal action
$Z-\norm{\bbeta}^2$.

\end{corollary}

A single budget must solve all these conditional problems with the same cap. Even with a large bias, a budget that always binds need not be optimal, because the principal may prefer different pooled choices on different lines.

For two independent decisions, this yields a simple sufficient condition for budget delegation:
\begin{corollary}\label{cor:symmetric_logconcave}
Suppose $n=2$, $b(a)=-\norm{a}^2/2$, $g(s)=s-(\beta,\beta)$ for a scalar $\beta>0$, and $S=[\underline{s},\overline{s}]^2$, and let the two coordinates of the state be independent draws from a common log-concave density that is symmetric about the midpoint $m=(\underline{s}+\overline{s})/2$. 
If \(\beta\geq\frac{\overline{s}-\underline{s}}{2}\), then budget delegation is optimal.
\end{corollary}

If \(\beta\geq\frac{\overline{s}-\underline{s}}{2}\), it can be verified that in the one-dimensional problem the principal would not give any discretion to the agent and optimally decide without consulting him. In the two-dimensional problem, budget delegation is then optimal. 

\subsection{The uniform example}

We return to our uniform example to illustrate. Recall that $b(a)=-\norm{a}^2/2$ and $s$ is uniformly distributed on $[0,1]^2$. \autoref{fig:intro_example} shows the optimal delegation sets in the two bias regimes, which we now derive.

\textbf{Large bias ($\beta\geq1/2$).}
At these bias levels separate delegation would eliminate all discretion, and applying \autoref{cor:symmetric_logconcave} implies that budget delegation is optimal. To illustrate the conditions of \autoref{thm:main_result}, we use it to argue directly that budget delegation with the budget set $a_1+a_2\leq1-2\beta$ is optimal. For this delegation set, the agent chooses his perpendicular projection onto its boundary, so types are pooled along segments parallel to $(1,1)$. Recall that $\nu$ equals one in the interior of the square, $\beta$ on the west and south edges, and $-\beta$ on the north and east edges. Since these segments are parallel, the conditional interior density is constant along each of them. The positive and negative endpoint weights cancel in total mass, but they shift the signed barycenter southwest by $(\beta,\beta)$ relative to the segment's midpoint. All these midpoints lie on $s_1+s_2=1$, so the signed barycenters lie exactly on the budget line and each assigned action is its pool's signed barycenter. The remaining convex-order comparisons required by \autoref{thm:main_result} hold because $\nu$ is negative only on the north and east edges and the budget line lies south-west of the type space.\footnote{Indeed, the segment from the assigned action to the pool's north-east endpoint contains the pool within $S$. On this segment, a convex function lies below the affine function that agrees with it at the two ends. Since the only negative mass of $\mu|_Q$ is at the north-east endpoint, the convex function has the smaller integral against $\mu|_Q$.}

\textbf{Small bias ($0<\beta<1/2$).}
Here the budget line defined above passes through the type space and the convex-order conditions fail. We construct instead a curved boundary, giving types in a southwest region their full-discretion payoff. By symmetry it suffices to describe the branch above the diagonal.

To determine the boundary of the delegation set, we parametrize it by $t\mapsto a(t)$, where $a(t)$ is the action that type $(t,1)$ chooses. Consider first a smooth decreasing portion with $a(t)$ in $[0,1]^2$ and $a_1'(t)>0$. Since the agent picks the closest point of the delegation set, the types pooled at $a(t)$ form the segment from $a(t)$ to $(t,1)$, and this segment is normal to the boundary:
\begin{align}\label{eq:uniform_normal}
 a'(t)\cdot\bigl[(t,1)-a(t)\bigr]=0.
\end{align}
By \autoref{thm:main_result}, $a(t)$ must be the signed barycenter of the conditional virtual measure on that segment.

\begin{figure}[tbp]
\centering
\begin{tikzpicture}[x=1cm,y=1cm,
  font=\footnotesize, line cap=round, line join=round,
  >=Stealth,
  bandguide/.style={black!38,line width=.4pt,densely dashed},
  banddim/.style={<->,black!75,line width=.5pt},
  bandnormal/.style={black!80,line width=.85pt},
  bandfrontier/.style={blue!55!black,line width=1.15pt},
  every node/.style={inner sep=2pt}]
\definecolor{bandblue}{RGB}{38,94,148}
\definecolor{bandred}{RGB}{176,66,48}

\coordinate (A) at (1.1,1.4);
\coordinate (Ap) at (2.15,.77);
\coordinate (N) at (3.2,4.4);
\coordinate (Np) at (3.965,4.4);
\coordinate (B) at (2.465,1.4);
\fill[bandblue!13] (A)--(N)--(Np)--(B)--cycle;
\node[bandblue] at (2.72,3.05) {$T_\varepsilon$};
\fill[bandblue!30] (A)--(B)--(Ap)
  plot[domain=2.15:1.1,samples=35]
  (\x,{1.4-.7*(\x-1.1)+(2/21)*(\x-1.1)^2}) --cycle;
\draw[black!50,line width=.6pt] (.25,4.4)--(5.25,4.4);
\node[anchor=south west,text=black!70] at (.25,4.44) {north edge};
\draw[bandred,line width=2pt] (N)--(Np);
\draw[bandnormal] (A)--(N);
\draw[bandnormal] (Ap)--(Np);
\draw[bandfrontier,bandblue,domain=.55:3.25,samples=60]
  plot (\x,{1.4-.7*(\x-1.1)+(2/21)*(\x-1.1)^2});
\node[bandblue,anchor=west] at (3.28,.33) {boundary};

\draw[bandguide] (N)--(3.2,5.02);
\draw[bandguide] (Np)--(3.965,5.02);
\draw[banddim] (3.2,4.9)--node[above] {$\varepsilon$}(3.965,4.9);
\node[anchor=south east] at (3.14,4.42) {$(t,1)$};
\node[anchor=south west] at (4.01,4.42) {$(t+\varepsilon,1)$};
\draw[bandred,line width=.45pt] (3.65,4.36)--(4.38,3.94);
\node[anchor=west,bandred] at (4.38,3.94) {$-\beta\varepsilon$};

\draw[bandguide] (-.25,4.4)--(.22,4.4);
\draw[bandguide] (-.25,1.4)--(A);
\draw[banddim] (-.14,1.4)--node[left] {$1-a_2(t)$}(-.14,4.4);
\draw[bandblue,line width=.7pt] (A)--(B);
\draw[banddim] (1.15,1.4)--node[above=3pt,xshift=4pt] {$w_\varepsilon$}(2.415,1.4);
\fill[bandblue] (A) circle (1.6pt);
\fill[bandblue] (Ap) circle (1.6pt);
\draw[bandblue,fill=white,line width=.6pt] (B) circle (1.4pt);
\node[anchor=north east] at (1.04,1.32) {$a(t)$};
\node[anchor=north] at (2.15,.67) {$a(t+\varepsilon)$};
\draw[black!55,line width=.45pt] (1.9,1.19)--(3.22,1.92)--(3.55,1.92);
\node[anchor=west] at (3.55,1.92) {$O(\varepsilon^2)$};

\end{tikzpicture}
\caption{The thin band between neighboring pools.}
\label{fig:small_bias_band}
\end{figure}

Consider the thin band $T_\varepsilon$ between the pools labelled $t$ and $t+\varepsilon$, shown in \autoref{fig:small_bias_band}. The pool labelled $t'$ leaves the north edge at $(t',1)$ and is normal to the boundary at $a(t')$. Simple geometric arguments then yield that the horizontal width of $T_{\varepsilon}$ at the height of $a(t)$ is
\begin{align}\label{eq:uniform_band_width}
 w_\varepsilon=\varepsilon+[1-a_2(t)]
 \left[\frac{a_2'(t+\varepsilon)}{a_1'(t+\varepsilon)}
       -\frac{a_2'(t)}{a_1'(t)}\right],
\end{align}
while at the north edge the width is $\varepsilon$. Up to a region of area $O(\varepsilon^2)$, $T_\varepsilon$ is therefore a trapezoid of vertical extent $1-a_2(t)$ with parallel sides $w_\varepsilon$ and $\varepsilon$. Its interior carries virtual mass $\tfrac12[1-a_2(t)](w_\varepsilon+\varepsilon)+O(\varepsilon^2)$, and the vertical coordinate of its barycenter is
\[
 a_2(t)+[1-a_2(t)]\frac{w_\varepsilon+2\varepsilon}{3(w_\varepsilon+\varepsilon)}
 +O(\varepsilon).
\]
The north-edge segment carries virtual mass $-\beta\varepsilon$ at height one. 
Balancing the vertical first moments at the height of $a(t)$ and dividing by $1-a_2(t)$ gives
\begin{align}\label{eq:uniform_band_moment}
 \frac{1-a_2(t)}{6}(w_\varepsilon+2\varepsilon)
 =\beta\varepsilon+O(\varepsilon^2).
\end{align}
Dividing by $\varepsilon$ and letting $\varepsilon$ vanish, the boundary solves
\begin{align}\label{eq:uniform_boundary_ode}
 [1-a_2(t)]\left(3+[1-a_2(t)]\frac{\mathrm d}{\mathrm dt}
                    \frac{a_2'(t)}{a_1'(t)}\right)=6\beta.
\end{align}
Thus the barycenter condition yields a second-order differential equation that determines how the slope of the boundary changes. \autoref{app:uniform} fixes the two constants of its solution by imposing that the boundary condition holds at the pool that collapses to the northwest corner, and that it meet its reflection on the diagonal. It also verifies that the convex-order conditions are satisfied and the delegation set is optimal.

Two features are worth noting. First, wherever $1-a_2(t)<2\beta$, the equation requires the slope $a_2'/a_1'$ to increase with $t$, and hence with $a_1$, which implies that the delegation set is not convex. The appendix establishes the stronger conclusion that no closed convex delegation set is optimal for $0<\beta<1/2$. Second, the boundary also has a portion outside the square, where a pool reaches the west edge before reaching its assigned action. Therefore, some types near the northwest corner choose an action with a negative first coordinate.

\subsection{When is randomization unnecessary?}\label{sec:boundary_regularity}

This subsection uses \autoref{thm:main_result} to give restrictions on where an optimal mechanism can randomize and where the action it assigns can jump. Both are governed by where virtual values are negative. Randomization and jumps are confined to the convex hull of the types whose utility the principal would like to hold down, so when those types lie on a single face of the type space, an optimal mechanism is deterministic and the corresponding delegation set is connected.

Common regularity conditions require that the bias and density do not vary too quickly, so that virtual values $\nu$ are positive in the interior of $S$. Write $K=\conv(\supp\mu^-)$ for the convex hull of the support of the negative part of $\mu$. 

\begin{proposition}\label{prop:negative_boundary_hull}
Suppose $\nu>0$ in $\interior S$. Then every optimal indirect utility $U$ is continuously differentiable on $\interior S\setminus K$, and each type in this region receives the deterministic action $\nabla U(s)$, which varies continuously with the type. Thus, up to a set of types of probability zero, randomization can occur only in $K$.
\end{proposition}
  
To see why, keep $U$ fixed on the types that carry negative virtual value and raise it everywhere else as far as convexity and the bound $h$ permit. This does not lower the principal's payoff, so by optimality it leaves her payoff unchanged; since $\nu$ is strictly positive in the interior, an optimal $U$ therefore already coincides with the raised function. That function is built from the fixed values together with points on the graph of $h$. A type outside $K$ cannot be reached from the fixed values alone, so a plane supporting the raised function there must also touch $h$ and hence correspond to a deterministic action. For a convex function, differentiability on an open set also makes the gradient continuous, which gives the continuous action rule.

If $K$ is contained in a proper face of $S$, we get the stronger result:
\begin{corollary}\label{cor:negative_face_connected}
Suppose $\nu>0$ in $\interior S$ and $\supp\mu^-$ is contained in a proper face of $S$. Then every optimum is deterministic almost everywhere and has a continuous action rule on $\interior S$, and there is a connected optimal delegation set.
\end{corollary}

The corollary follows because a proper face does not meet the interior of $S$, so the proposition applies throughout a connected set of types, over which the chosen action varies continuously.

\subsection{Proof sketch}
\label{sec:proof_sketch}

To simplify the exposition, we assume throughout this subsection that the density vanishes on the boundary of the type space, so $f(s)=0$ for all $s\in\bd S$. Then $\nu$ vanishes on $\bd S$, so $\mu$ is absolutely continuous with respect to Lebesgue measure. Since a convex function is differentiable Lebesgue-almost everywhere, $U$ is differentiable at $|\mu|$-almost every type, and $\Q$ is simply the partition of $\R^n$ into the level sets of $\nabla U$: no choice has to be made at kinks.

To prove \autoref{thm:main_result}, we use duality in linear programming. To formulate the dual program, it is more convenient to work with indirect utilities that are defined on a compact domain.
But recall that it is not enough in \autoref{lemma:feasible_set} to only require that $U(s)\le h(s)$ for all $s\in S$ (see \autoref{fig:domain_indirect_util}). 
The following technical result ensures that we can restrict the indirect utilities to have a compact domain as long as this domain is chosen large enough.

\begin{lemma}\label{lemma:compact_domain}
There is a compact convex $X\subseteq\R^n$ such that the principal's problem can be written as $\max \{\int U \,\mathrm d\mu| U:X\rightarrow\R,\ U \text{ convex},\ U(x)\le h(x) \forall x\in X\}$.
\end{lemma}

Formally, we choose $X$ large enough that the restriction to $S$ of every solution of the problem on $X$ extends to a convex function on $\R^n$ that still lies below $h$; by \autoref{lemma:feasible_set} this extension is a feasible indirect utility with the same payoff. The choice of $X$ has two parts. First, $X$ contains a neighborhood of $S$ and of $g(S)$, on which a solution is bounded above by $h$; and since a solution does at least as well as always choosing the action $0$, it cannot be too low on $S$ either. Together these two bounds limit how steep a solution can be on $S$. Second, for each slope $a$ that remains possible, $X$ contains the type whose ideal action is $a$. This is the type at which a plane of slope $a$ comes closest to $h$, so checking the plane there places it below $h$ everywhere; taking the supremum of the supporting planes of a solution at types in $S$ then gives the required extension.

Now let $X$ be as in the above lemma and denote by $\U$ the set of convex continuous functions that map $X$ to $\R$ and by $\M_+$ the set of positive measures on $X$. We can formulate the principal's problem as follows:
\begin{align*}\label{eq:primal_problem}
    &\max_{U \in\U} \int U(s) \,\mathrm d\mu(s) \tag{P}\\
    & \text{s.\,t. } U\le h
\end{align*}

\paragraph{The dual problem}
We will show that the following problem is the dual problem:
\begin{align}\label{eq:dual_problem}
    &\inf_{\gamma \in \M_+} \int h(s) \,\mathrm d\gamma(s) \tag{D}\\
    & \text{s.\,t. }  \gamma \ge_{cx} \mu.\nonumber
\end{align}

Note that $h$ is a convex function; therefore, if $\mu$ was a positive measure, this would be a trivial problem with solution $\gamma= \mu$. However, since $\mu$ is a signed measure and $\gamma$ has to be a positive measure, $\mu$ is not feasible in general.

It is easy to see that weak duality holds, i.e., that the value of the primal problem \eqref{eq:primal_problem} is always below the value of the dual problem \eqref{eq:dual_problem}. Indeed, for feasible $U$ and $\gamma$,
\begin{align}\label{eq:weak_duality}
\int U(s) \,\mathrm d\mu(s)\underbrace{\le}_{(i)} \int U(s)\,\mathrm d \gamma(s)\underbrace{\le}_{(ii)} \int h(s) \,\mathrm d \gamma(s)
\end{align}
 since (i) $U$ is convex and $\mu\le_{cx} \gamma$ and (ii)  $\gamma$ is a positive measure and $U\le h$. The following result shows that strong duality holds, that is the optimal values of both problems are equal and the dual problem has a solution. 

\begin{lemma}[Strong duality]\label{l:strong_duality}
A feasible indirect utility $U$ is optimal if and only if there exists a positive measure $\gamma\ge_{cx}\mu$ such that 
\begin{align}
    U(s)&=h(s) \text{ for $\gamma$-almost every\ $s$ } \label{eq:cs1} \\
    \int U(s) \,\mathrm d\mu(s) &= \int U(s) \,\mathrm d\gamma(s). \label{eq:cs2}
\end{align}
\end{lemma}

This result is an analogue of a result in the revenue-maximization problem of a multiproduct monopolist \parencite[see Theorem 2 in][]{daskalakis-etal2017}. Our formulation of the delegation problem allows us to easily deduce strong duality. Note that there is a convex function $U$ such that $h(x)-U(x)>0$ for all $x\in X$. Therefore, Slater's constraint qualification is satisfied and standard results in linear programming imply that the dual problem has a solution and that the optimal solutions of the primal and dual problems achieve the same value. Since both inequalities in \eqref{eq:weak_duality} have to hold as equalities, \autoref{l:strong_duality} follows.

\paragraph{Proof idea for \autoref{thm:main_result}.} Sufficiency does not use duality. Applied to any feasible indirect utility $V$, the conditions of the theorem give $\int V \,\mathrm d\mu|_Q\le \int U \,\mathrm d\mu|_Q$ for almost every $Q\in\Q$, exactly as in the discussion following the theorem; integrating over partition elements yields $\int V \,\mathrm d\mu\le \int U \,\mathrm d\mu$.

For the converse direction, suppose $U$ is optimal. By \autoref{l:strong_duality}, there is a positive measure $\gamma$ satisfying the complementary slackness conditions \eqref{eq:cs1} and \eqref{eq:cs2} and $\gamma \ge_{cx} \mu$; this last condition is equivalent to $\gamma + \mu^- \ge_{cx} \mu^+$, where $\mu^+$ and $\mu^-$ denote the positive and negative parts of $\mu$, respectively. Strassen's theorem then implies that $\gamma+\mu^-$ is a mean-preserving spread of $\mu^+$: one can obtain the measure $\gamma+\mu^-$ by taking, for every $s$, the mass $\mu^+$ puts on $s$ and spreading it according to a probability measure $D_s$ with expected value $s$. Since $U$ is convex, Jensen's inequality implies $U(s)\le \int U(x) \,\mathrm dD_s(x)$, and \eqref{eq:cs2} forces equality. But equality
 means that $U$ is affine on the convex hull of the support of $D_s$.

We check that the spread does not move mass between partition elements: Since the spread produces the measure $\gamma+\mu^-$, it suffices that $U$ be differentiable at $\mu^-$-almost every point and at $\gamma$-almost every point. The first holds because $\mu^-$ is absolutely continuous and a convex function is differentiable almost everywhere. The second holds because \eqref{eq:cs1} gives $U(x)=h(x)$ for $\gamma$-almost every $x$, and a convex function that touches the differentiable $h$ from below is differentiable at the point of contact. At such a point $x$ in the support of $D_s$, the supporting plane that makes Jensen's inequality an equality at $s$ touches $U$ both at $s$ and at $x$, so $\nabla U(x)=\nabla U(s)$ and $x$ lies in the same partition element as $s$.

Conditioning $\gamma$ and $\mu$ on the partition elements therefore yields, for almost every $Q\in\Q$, a positive measure $\gamma|_Q$ with $\gamma|_Q\ge_{cx} \mu|_Q$ and $\gamma|_Q(Q)=\mu|_Q(Q)$; in particular $\mu|_Q(Q)\ge 0$. By \eqref{eq:cs1}, $\gamma|_Q$ is concentrated on the types in $Q$ at which $U$ meets $h$, and there is at most one such type because $U$ is affine on $Q$ and $h$ is strictly convex. So $\gamma|_Q$ is a point mass or the zero measure, that is, $\gamma|_Q=\delta_Q$, which is the assertion of the theorem. The Appendix carries out these steps without assuming that $f$ vanishes on the boundary of $S$.

\section{More general payoffs for the principal}\label{sec:general_payoffs}

In this section, we consider a more general specification of the principal's preferences and assume the principal's payoff from action $a$ in state $s$ is $u_P(a,s)$. 
We assume that \(u_P\) is continuous, that \(u_P(\cdot,s)\) is concave and differentiable for each \(s\), and that the partial derivatives of \(\nabla_a u_P\) with respect to \(s\) exist and are jointly continuous.

Our earlier specification that the principal's preferences satisfy \autoref{ass:principal_pref} ensures that the principal's expected payoff is completely determined by the agent's indirect utility $U$ and, moreover, is a linear function of $U$. For the more general preferences considered here, neither conclusion holds in general. To illustrate the first issue, suppose the agent has quadratic preferences and that $m$ is an incentive-compatible mechanism that imposes nondegenerate lotteries for a positive measure of types. Then every mechanism that imposes, for every type, a lottery with the same expected value and variance as $m$ induces the same indirect utility as $m$, but the principal will not be indifferent between these mechanisms in general. Hence, to find an optimal mechanism, in addition to finding an optimal indirect utility $U$, one would have to optimize over which lotteries are used to induce $U$. However, for deterministic mechanisms the indirect utility still determines the principal's expected payoff even with the more general preferences (since the gradient of $U$ almost everywhere determines the deterministic action chosen by the corresponding mechanism). We use this insight to provide sufficient conditions under which a given deterministic mechanism is optimal. We then show that when the principal's preferences are quadratic with different weights across dimensions, the best lotteries inducing a given indirect utility can be identified explicitly, so that the problem is again one over indirect utilities. 

\paragraph{Deterministic mechanisms.}
To state the result, let 
\[\kappa:=\sup\{\kappa'\in \R: u_P(a,s) -\kappa' b(a) \text{ is concave in $a$ for all $s$}\}\]
be a measure of concavity for the principal's payoff relative to the agent's payoff; the supremum is attained, since concavity is preserved under pointwise limits, and we assume it is finite.

For a feasible indirect utility $U$, let
\begin{align}\label{eq:general_payoff_bound}
\Phi(U):=\int \Big[u_P(\nabla U(s),s) + \kappa\big[U(s)-u_A(\nabla U(s),s)\big]\Big] \,\mathrm dF(s).
\end{align}
Every mechanism inducing $U$ gives the principal at most $\Phi(U)$: at almost every type its lottery has mean $\nabla U(s)$, so applying Jensen's inequality to the concave function $u_P(\cdot,s)-\kappa b(\cdot)$ and substituting $\E[b]=U(s)-\nabla U(s)\cdot s$ for the remaining term leaves exactly the bracket in \eqref{eq:general_payoff_bound}. Since $u_P(\cdot,s)-\kappa b(\cdot)$ is concave and $U\mapsto\nabla U$ is linear, $\Phi$ is concave. For a deterministic mechanism, $U(s)=u_A(\nabla U(s),s)$, so $\Phi(U)$ is the principal's expected payoff. Such a mechanism is therefore optimal whenever its indirect utility maximizes $\Phi$ among feasible indirect utilities.

Given a feasible indirect utility $U$, let $\tilde{\mu}$ denote the finite signed measure on $S$ satisfying
\begin{align}\label{eq:def_tilde_mu}
\int V(s) \,\mathrm d\tilde{\mu}(s) = \int \Big[\kappa V(s) + \big[\nabla_a u_P(\nabla U(s),s) -\kappa s - \kappa \nabla b(\nabla U(s))\big]\cdot \nabla V(s)\Big] \,\mathrm dF(s)
\end{align}
for every continuously differentiable $V$. Such a measure exists,
and \eqref{eq:def_tilde_mu} continues to hold when $V$ is a difference of feasible indirect utilities; in particular $U$ need not be differentiable, since its gradient enters only $F$-almost everywhere.

If the principal's payoff satisfies \autoref{ass:principal_pref}, then $\kappa=1$ and $\tilde{\mu}$ equals $\mu$, independently of $U$. With the more general payoff, $\tilde{\mu}$ measures how, starting at indirect utility $U$, marginally increasing the agent's payoff changes the principal's payoff.

To state conditions under which a deterministic mechanism with indirect utility $U$ is optimal, let $\Q$ be the partition induced by $U$ as in \autoref{sec:optimal_mechanisms}; note that it depends on $U$ alone. For any $Q\in\Q$, if there is $s\in Q$ with $U(s)=h(s)$ we denote by $\delta_Q$ a point mass of mass $\tilde{\mu}|_Q(Q)$ at $s$. If $U(s)<h(s)$ for all $s\in Q$, let $\delta_Q$ denote the zero-measure.

\begin{theorem}\label{thm:main_result_nonlinear}
Let $U$ be a feasible indirect utility induced by a deterministic mechanism $m$. Then $m$ is optimal if for $|\tilde{\mu}|$-a.e.\ $Q\in \Q$, $\tilde{\mu}|_Q(Q)\ge 0$ and $\tilde{\mu}|_Q\le_{cx} \delta_Q$.
\end{theorem}

To prove \autoref{thm:main_result_nonlinear}, we compare $m$ with an arbitrary incentive-compatible mechanism $\tilde{m}$ with indirect utility $\tU$. Since $u_P(a,s)-\kappa b(a)$ is concave in $a$, replacing it by its tangent at the action chosen under $m$ yields
\[ \int \big[u_P(\tilde{m}(s),s) - u_P(m(s),s)\big] \,\mathrm dF(s) \le \int (\tU-U) \,\mathrm d\tilde{\mu}. \]
Indeed, $\Phi(\tU)$ bounds the principal's payoff from $\tilde{m}$, $\Phi(U)$ equals her payoff from $m$, and $\Phi(\tU)-\Phi(U)\le\int(\tU-U)\,\mathrm d\tilde{\mu}$ by concavity of $\Phi$. It therefore suffices that $U$ maximize $\int \tU \,\mathrm d\tilde{\mu}$ among feasible indirect utilities, and the conditions of the theorem deliver this by the same comparison on each partition element that establishes sufficiency in \autoref{thm:main_result}. Only that comparison is used: the necessity part of \autoref{thm:main_result} rests on properties of $\mu$ that $\tilde{\mu}$ need not share.\footnote{Concavity enters this argument only through the tangent inequality, and only at the action the mechanism assigns. Since $u_P(a,s)-\kappa b(a)$ and $u_P(a,s)-\kappa u_A(a,s)$ differ by a term that is linear in $a$, \autoref{thm:main_result_nonlinear} therefore also holds whenever $a\mapsto u_P(a,s)-\kappa u_A(a,s)$ admits a concave upper bound that is tight at the action assigned to $F$-almost every type, with the coefficient of $\nabla V$ in \eqref{eq:def_tilde_mu} read as the gradient of that bound at the assigned action. We use this form in \autoref{sec:outside_option}, where a fixed cost makes payoffs discontinuous, so that the difference need not be concave for any $\kappa$.}

Since $\tilde{\mu}$ coincides with $\mu$ from \autoref{sec:optimal} if the principal's preferences take the form in \autoref{ass:principal_pref}, it follows that these conditions are necessary in that case. Also, for one-dimensional types and general payoffs, one can verify that the conditions of \autoref{thm:main_result_nonlinear} specialize for interval delegation to the sufficient conditions obtained in \citet{AB:13}.

\paragraph{Differently weighted quadratic preferences.}
In general, a stochastic mechanism need not attain the bound $\Phi(U)$. With differently weighted quadratic preferences, however, every feasible $U$ is induced by a mechanism that attains it. Suppose
\[ b(a)=-\frac{1}{2}\norm{a}^2, \qquad u_P(a,s)=a\cdot g(s)-\frac{1}{2}\,a\cdot Wa, \]
for a fixed symmetric positive definite matrix $W$. The principal's ideal action is $W^{-1}g(s)$ and $W=I$ recovers \autoref{ass:principal_pref}. Here $u_P-\kappa' b$ is concave in $a$ exactly when $W-\kappa' I$ is positive semidefinite, so $\kappa>0$ is the smallest eigenvalue of $W$.

For a feasible $U$, any mechanism inducing $U$ assigns a lottery with mean $\nabla U(s)$ and total variance $\sigma^2(s):=2[u_A(\nabla U(s),s)-U(s)]$, twice the agent's utility loss from randomization at $s$; the indirect utility pins down these two moments but not the direction in which the variance is spread. Changing that direction leaves the mean and the expected squared norm of the lottery unchanged, hence also every type's evaluation of every report, so the principal is free to place the required variance where she penalizes it least.

\begin{proposition}\label{prop:weighted_quadratic}
For every feasible $U$, the highest expected payoff among mechanisms inducing $U$ equals $\Phi(U)$. It is attained by assigning, at almost every type, the two actions $\nabla U(s)\pm\sigma(s)\,e$ with equal probabilities, where $e$ is a unit eigenvector of $W$ for the eigenvalue $\kappa$.
\end{proposition}

Thus, for every feasible $U$, $\Phi(U)$ is the highest payoff the principal can obtain from a mechanism inducing $U$. Expanding \eqref{eq:general_payoff_bound} gives
\begin{align}\label{eq:weighted_quadratic_objective}
\Phi(U)=\int \Big[\kappa U(s)+\big[g(s)-\kappa s\big]\cdot \nabla U(s)-\frac{1}{2}\,\nabla U(s)\cdot (W-\kappa I)\nabla U(s)\Big] \,\mathrm dF(s).
\end{align}
The principal's problem is to maximize this concave functional over convex $U:\R^n\rightarrow\R$ with $U\le h$. Its quadratic term is negative semidefinite; it vanishes when $W$ is a positive multiple of the identity, which returns the linear problem of \autoref{sec:optimal}, and when $W$ is diagonal an optimal mechanism randomizes only in a least-weighted dimension.

Since $\tilde{\mu}$ is the derivative of $\Phi$ at $U$, a feasible $U$ is optimal if and only if $\int (V-U) \,\mathrm d\tilde{\mu}\le 0$ for every feasible $V$. In particular, the conditions of \autoref{thm:main_result_nonlinear} are sufficient for a stochastic mechanism as well, provided it uses the lotteries of \autoref{prop:weighted_quadratic}.

\paragraph{Biased and aligned decisions.}
Randomization lowers the agent's utility and works like money burning, but it is costly for the principal as well: under \autoref{ass:principal_pref}, her cost of a given risk equals the agent's. Something that is specific to multidimensional problems is that the principal can place the risk in a decision she cares little about, while the agent still bears its full cost. The following makes this precise.

Write the state as $s=(x,z)\in S\subseteq\R^k\times\R^{n-k}$ and the action as $a=(a^x,a^z)$, with $1\leq k<n$. The first $k$ decisions are biased and the remaining $n-k$ decisions are aligned: the agent's ideal action is $s$, while the principal's payoff is
\[
 u_P(a,s)=-\frac12\norm{a^x-(x-\bbeta)}^2-\frac{\kappa}{2}\norm{a^z-z}^2,\qquad 0<\kappa<1,
\]
where $\bbeta\in\R^k\setminus\{0\}$ collects the possibly unequal biases. The principal thus agrees with the agent on the aligned decisions but cares less about them. Up to a term that does not depend on $a$, this is the weighted quadratic payoff above with weight $1$ on the biased and weight $\kappa$ on the aligned decisions.

\begin{proposition}\label{prop:aligned_randomization}
Fix a nonzero bias vector $\bbeta$. There exists $\kappa_0>0$ such that,
whenever $0<\kappa<\kappa_0$, every optimal mechanism uses
nondegenerate lotteries on a positive-probability set of types.
An optimal mechanism can concentrate all randomization
in a single aligned decision.
\end{proposition}

No independence assumption is required for the proposition.
If the biased and aligned states are independent, an optimal
mechanism can additionally be chosen so that
$\E[a^z\mid x,z]=z$. The expected aligned action is thus undistorted, and randomization in the aligned decisions provides incentives in the biased decisions.

Randomization here does what money burning does when transfers are
available, except that the instrument is a decision that has to be
taken in any case. Placing the required variance $\sigma^2(s)$ in an
aligned decision costs the agent $\sigma^2(s)/2$ but costs the principal
only $\kappa\sigma^2(s)/2$, so incentives become arbitrarily cheap
relative to their value as $\kappa$ falls. Deterministic distortions of the aligned decisions would be equally cheap. However, in a deterministic mechanism the action is $\nabla U$, so a distortion of the aligned decisions that varies with $x$ forces the biased decisions to vary with $z$, whereas the risk of a lottery can vary with $x$ freely.
  
\section{Outside option for the agent}\label{sec:outside_option}

In some economic applications, outside options for the agent arise naturally. As examples, we discuss below applications to veto bargaining and monopoly regulation. Often, these outside options can be formulated as a simple condition on the indirect utility of the agent, and the indirect utility approach extends easily. 

To model this outside option, consider our general model above and suppose $h\ge 0$ and every type of the agent has an outside option that yields utility $0$. A mechanism is \emph{feasible} for this model if it is incentive compatible and yields each type $s\in S$ a nonnegative indirect utility: $u_A(m(s),s)\ge 0$ for all $s\in S$. Clearly, any feasible mechanism induces an indirect utility that satisfies $U(s)\ge 0$ for all $s\in S$. Conversely, the arguments used to prove \autoref{lemma:feasible_set} can be used to prove the following (we omit the straightforward argument).
\begin{lemma}\label{l:ic_outside}
An indirect utility $U$ corresponds to a feasible mechanism for the model with outside option if and only if $U$ is convex, $U(s)\ge 0$ for all $s\in S$, and $U\le h$. 
\end{lemma}

Let $\U_0:=\{U\in \U:\ \forall s\in S,  U(s)\ge 0\}$ denote the subset of convex functions that are nonnegative for all types in $S$. Under \autoref{ass:principal_pref}, the principal's problem is then
\begin{align*}
&\max_{U\in\U_0} \int U \,\mathrm d\mu \\ 
&\text{s. t. } U\le h.
\end{align*}

An analogue of \autoref{thm:main_result} holds in this setting, with two modifications. First, the convex order $\le_{cx}$ is replaced by the order $\le_{\U_0}$, where $\alpha\le_{\U_0} \beta$ if and only if $\int V \,\mathrm d\alpha \le \int V \,\mathrm d\beta$ for all $V\in\U_0$. Second, the mass of $\delta_Q$ is no longer determined. Since the outside option prevents the principal from lowering the utility of types with $U(s)=0$, the dual solution may put more mass on a contact point than $\mu|_Q(Q)$. In this section we therefore let $\delta_Q$ denote any nonnegative measure concentrated on $Q\cap\{s:U(s)=h(s)\}$, with its mass left free; as before there is at most one such contact point, so $\delta_Q$ is a point mass or the zero measure.

\begin{theorem}\label{thm:outside_option}
Consider the model with outside option and let $U$ be a feasible indirect utility. Then $U$ is optimal if and only if for $|\mu|$-a.e. $Q\in \Q$ there is such a measure $\delta_Q$ with
\[ \mu|_Q\le_{\U_0} \delta_Q \qquad\text{and}\qquad \int U \,\mathrm d\mu|_Q = \int U \,\mathrm d\delta_Q. \]
\end{theorem}

In \autoref{thm:main_result} the second condition was automatic: the convex order compares affine functions in both directions and thereby fixes the mass of $\delta_Q$ at $\mu|_Q(Q)$. The order $\le_{\U_0}$ does not, so the equality must be imposed, and it is what now pins the mass down: if $Q$ contains a contact point $s$ with $h(s)>0$, that mass is $\int U \,\mathrm d\mu|_Q/h(s)$.

Sufficiency follows from the same payoff comparison as in \autoref{thm:main_result}, and that comparison uses no property of $\mu$ beyond the two displayed conditions; we apply it with $\tilde{\mu}$ in place of $\mu$ in the monopoly application below. For necessity one has in addition to account for the types whose outside option binds; the argument is given in the appendix and again relies on strong duality. The compact-domain reduction of \autoref{lemma:compact_domain} applies here as well, since extending an indirect utility preserves its values on $S$ and hence the outside-option constraint, and on this domain the dual problem is to minimize $\int h \,\mathrm d\gamma$ over positive measures $\gamma$ satisfying $\mu\le_{\U_0}\gamma$. The proof of \autoref{l:strong_duality} does not apply, however, because Slater's constraint qualification is not necessarily satisfied (there might be no $U\in\U_0$ with $U<h$). Strong duality can still be established by showing that the value function has bounded slope.

\begin{lemma}\label{l:strong_duality_outside}
Strong duality holds for the problem with outside options, and the dual problem has a solution.
\end{lemma}

\paragraph{Veto bargaining}

To illustrate this model variant, consider a model of incomplete information veto bargaining in which a proposer can make a policy proposal which the veto player can either accept or reject \citep[see, e.g.,][]{kartik21}. After a rejection, a status quo policy remains in effect. The principal and agent have quadratic preferences over the policy that is implemented, where the principal's ideal policy is known to be $a_P$ and the agent's ideal policy $s$ is his private information. Since the principal does not know which policies are acceptable to the veto player---that is, which policies the veto player prefers to the status quo policy---the principal generally benefits from eliciting the agent's information by committing to a mechanism.

This model can be interpreted as a delegation problem in which the agent has an outside option of obtaining the status quo and where the principal's preferences are independent of the state. Let $u_A(a,s) = a\cdot s -\frac{1}{2}\norm{a}^2$ denote the agent's payoff if policy $a$ is implemented in state $s$ and $u_P(a)=a\cdot a_P-\frac{1}{2}\norm{a}^2$ the principal's payoff. Let $a_{SQ}=0$ denote the status quo action. Then any mechanism satisfying the outside option has an indirect utility satisfying $U(s)\ge 0$ for all $s\in S$ and any $U\in \U_0$ with $U\le h$ corresponds to some feasible mechanism.

The indirect utility approach can be used to solve the one-dimensional version of this problem analyzed in \citet{kartik21}. This approach simplifies their arguments because it directly yields necessary and sufficient conditions. More importantly, the indirect utility approach is a useful tool to analyze incomplete-information veto bargaining over multidimensional policies.

\paragraph{Monopoly regulation with outside option.}
The method also yields new sufficient conditions in the canonical one-dimensional problem of regulating a privately informed monopolist without transfers \citep{AB:21}. A firm observes its productivity $s\in S=[\underline{s},\overline{s}]$ and chooses output $a\geq0$. Its profit and the regulator's payoff are
\[
 u_A(a,s)=sa+b(a)-c\mathbf{1}_{\{a>0\}},
 \qquad u_P(a,s)=u_A(a,s)+v(a),
\]
where $c\geq0$ is a fixed production cost, $v$ is weighted consumer surplus, and $b(0)=v(0)=0$. The firm can shut down for a payoff of zero. Feasibility is characterized as in \autoref{lemma:feasible_set}, except that the fixed cost changes which lotteries implement a given supporting line: a lottery between shutdown and a positive output can be needed even when expected output is small. The regularity assumptions and formal results are in \autoref{app:monopoly}.

A \emph{cap allocation} offers the menu $\{0\}\cup[\underline{a},\infty)$: a producing firm must supply at least $\underline{a}$, which corresponds to a price cap. For the nondegenerate caps considered in the appendix, types below $s_1=[c-b(\underline{a})]/\underline{a}$ shut down, types between $s_1$ and $s_2=-b'(\underline{a})$ produce $\underline{a}$, and higher types choose their unconstrained monopoly output. Expected profit is therefore zero on the shutdown region, affine on the pooling region, and equal to the full-discretion profit $h$ on the unconstrained region, as \autoref{fig:cap_allocation_2} illustrates.

A tighter price cap raises output among constrained firms but can also induce shutdown. When both cutoffs lie in the interior of $S$, the first-order condition for the output floor is
\begin{align}\label{eq:monopoly_cap_balance}
 \int_{s_1}^{s_2}
 \big[s-s_2+v'(\underline{a})\big]f(s)\,\mathrm ds
 =\frac{v(\underline{a})}{\underline{a}}
   (s_2-s_1)f(s_1).
\end{align}
The left-hand side is the marginal welfare effect of expanding output among the constrained firms. The right-hand side is the welfare lost from additional shutdowns: the marginal firm earns zero profit, so its exit costs the regulator the weighted consumer surplus $v(\underline{a})$.

Combining the payoff bound of \autoref{sec:general_payoffs} with the pooling comparisons above, \autoref{prop:optimality_cap_allocation} gives sufficient conditions for a particular cap to dominate every incentive-compatible mechanism, including lotteries over output and shutdown. The conditions compare alternatives on the unconstrained, pooling, and shutdown regions; in the interior-cutoff case, their equality on the pooling region is precisely \eqref{eq:monopoly_cap_balance}. They allow both output and the set of producing types to change, and the comparison is with a larger mechanism class than the deterministic quantity allocations studied by \citet{AB:21}.

\section{Conclusion}

Our paper studies a multidimensional model of optimal delegation. A key step to analyze the model is a simple characterization of incentive-compatible mechanisms. Using this characterization, we derive conditions that are necessary and sufficient for an arbitrary mechanism to be optimal and apply them to convex delegation sets and to budgets that link the decisions. Optimal mechanisms are deterministic when virtual values are positive in the interior of the type space and negative only on one of its faces. With more general preferences for the principal, however, randomization can be strictly optimal, because the principal can place the risk in decisions she cares little about. We also discuss how to incorporate outside options.

There are various directions that may be fruitful for future research. On the theoretical side, it would be of interest to find conditions that are simple to check and ensure that a mechanism in some class of specific mechanisms is optimal (for example, conditions on the primitives ensuring that some convex delegation set is optimal). Regarding applications, it would be of interest to understand how the multidimensional nature of information and actions affects conclusions that have been obtained in one-dimensional models. Examples could include bargaining over policies with rich policy spaces, or regulation of a monopolist that has private information about both fixed and marginal costs.

\begingroup
    \setlength{\emergencystretch}{3em}
    \section*{\phantomsection\refname}\addcontentsline{toc}{section}{\refname}
    \singlespacing
    \printbibliography[heading=none]
\endgroup 

\appendix

\section{Omitted proofs}
\newcommand{\tilU}{\tilde{U}}

\subsection{Defining the partition $\Q$}
\label{sec:partition}

Throughout, $U$ is a feasible indirect utility, that is, by \autoref{lemma:feasible_set}, a convex function $U:\R^n\rightarrow\R$ with $U\le h$. In particular, $\partial U(s)$ is a nonempty, compact, convex set for every $s\in\R^n$.

By the construction in the proof of \autoref{lemma:feasible_set}, each $p\in\partial U(s)$ is the expected action of a lottery that type $s$ can be assigned, and the partition element containing $s$ is determined by which of them is chosen. Where $U$ is differentiable there is only one, namely $\nabla U(s)$, and nothing has to be decided. At a kink we assign $s$ the subgradient that governs $U$ just inside $S$: fix $\hat{s}\in\interior S$ and let
\[
M_U(s):=\operatorname*{arg\,min}_{p\in\partial U(s)}p\cdot(s-\hat{s}),
\qquad
T_U(s):=\operatorname*{arg\,min}_{p\in M_U(s)}\norm{p}^2 .
\]
We then define
\[
Q_p:=\{s\in\R^n:T_U(s)=p\},
\qquad
\Q:=\{Q_p: Q_p\neq\emptyset\},
\]
so $\Q$ is a measurable partition of $\R^n$, and we write $Q(s)$ for the element containing $s$.

\subsection{Proofs for Section~\ref{sec:optimal}}

\begin{proof}[Proof of \autoref{lemma:compact_domain}]
It suffices to consider candidates with payoff at least $b(0)$, which is attainable by always choosing action $0$. Choose a compact convex neighborhood $K$ of $S\cup g(S)$ large enough to contain every point $2g(s)-t$ with $s,t\in S$. For any convex candidate $U\le h$ on a domain containing $K$, convexity gives
\[
b(0)\le \int U\,\mathrm d\mu
\le \int U(g(s))\,\mathrm dF(s)
\le \frac{U(t)+\max_K h}{2}
\qquad(t\in S).
\]
For the middle inequality, recall that the principal's payoff at $s$ is $U(s)+\nabla U(s)\cdot[g(s)-s]$, the value at $g(s)$ of the tangent to $U$ at $s$, which by convexity is at most $U(g(s))$. The last inequality uses that $g(s)$ is the midpoint of $t$ and $2g(s)-t$, both of which lie in $K$. Thus $U\ge 2b(0)-\max_K h$ on $S$. A convex function that is bounded below on $S$ and bounded above on a neighborhood of $S$ has uniformly bounded subgradients on $S$; since $K$ is such a neighborhood and $U\le h$ there, the supporting slopes on $S$ are uniformly bounded.

By continuity of $\nabla b$, we can choose a compact convex $X\supseteq K$ containing $-\nabla b(a)$ for every slope $a$ within this bound. Consider any such candidate $U$ on $X$. Evaluating a supporting plane of slope $a$ at $-\nabla b(a)$, where action $a$ is ideal, shows that its intercept is at most $b(a)$. Hence the plane lies below the payoff plane of action $a$, and therefore below $h$ everywhere. Taking the supremum of the supporting planes at points of $S$ gives a convex extension that agrees with $U$ on $S$ and lies below $h$ globally. By \autoref{lemma:feasible_set}, this extension is feasible and gives the same payoff. Conversely, every globally feasible indirect utility restricts to a feasible candidate on $X$, so the two problems have the same value.

Finally, these extensions are uniformly bounded and equi-Lipschitz on $X$. Applying Arzel\`a--Ascoli to the extensions of a maximizing sequence gives a uniformly convergent subsequence. Its limit is convex, lies below $h$, and attains the limiting payoff because $\mu$ is finite. Thus the maximum is attained.
\end{proof}

\begin{proof}[Proof of \autoref{l:strong_duality}]
    Let $\C(X)$ denote the vector space of continuous functions on $X$ with the supremum norm and recall that its dual space is the space of (Radon) measures on $X$, which we denote by $\M(X)$. Let $\V :=\{ g\in \C(X): \forall x\in X, g(x)\ge 0 \}$; the polar cones of $\U$ and $\V$ are defined by
\begin{align*}
    \U^* &:=\{ \gamma \in \M(X): \forall U\in\U, \int U \,\mathrm d\gamma \ge 0 \}\\
    \V^* &:=\{ \gamma \in \M(X): \forall g\in\V, \int g \,\mathrm d\gamma \ge 0 \}.
\end{align*}

The principal's problem can be written as $\max_{U \in \U} \int U\,\mathrm d\mu $ subject to $h-U\in \V$. 
This is a conical linear program and its dual is $\inf_{\gamma \in \V^*} \int h\,\mathrm d\gamma $ subject to $\gamma-\mu\in \U^*$ \parencite[e.g.,][]{shapiro10}.
Since $\V^*=\M_+(X)$ by the Riesz representation theorem \parencite[][p.\,265]{dunford-schwartz88a} and $\gamma-\mu\in \U^*$ is equivalent to $\gamma\ge_{cx} \mu$, \eqref{eq:dual_problem} is the dual problem. 

Since there is $U\in \U$ such that $h-U$ is in the interior of $\V$, Slater's condition is satisfied and standard results imply that strong duality holds \parencite[e.g.,][Proposition 2.8]{shapiro10}. 
Given \eqref{eq:weak_duality}, it follows that $U$ is optimal if and only if there is a positive measure $\gamma\ge_{cx} \mu$ such that $\int U \,\mathrm d\mu = \int U\, \mathrm d\gamma$ and $\int U\,\mathrm d\gamma = \int h\,\mathrm d\gamma$, which implies the result.
\end{proof}

As explained in the proof sketch in the main text, in the proof of \autoref{thm:main_result} we spread mass from a type $s$ to types $t$ that lie in the same affine piece of $U$ as $s$ and the argument requires that such $t$ belong to the same partition element as $s$. Where $U$ is differentiable this is immediate, since the supporting plane is then unique. At a kink the selection makes a choice, and two types that share a supporting plane need not be assigned the same one. The next lemma rules this out on a set of types that carries all the mass of $|\mu|$.

Let $G$ consist of the types $s\in\interior S$ at which $U$ is differentiable, together with the types $s\in\bd S$ at which $S$ has a unique outward normal and the restriction of $U$ to $\bd S$ is differentiable.

\begin{lemma}\label{l:partition}
The following statements hold.
\begin{enumerate}[label=(\roman*)]
\item $|\mu|(S\setminus G)=0$.
\item If $s,t\in G$ and $T_U(s)\in\partial U(t)$, then $Q(s)=Q(t)$.
\end{enumerate}
\end{lemma}

\begin{proof}
\textit{(i)} A convex function is differentiable Lebesgue-almost everywhere, so $U$ is differentiable at almost every type in $\interior S$. Since $S$ is convex, $\bd S$ is locally the graph of a convex function, which is differentiable almost everywhere; hence the outward normal is unique at $\H$-almost every boundary type. Similarly, $U$ is Lipschitz and therefore differentiable along $\bd S$ at $\H$-almost every boundary type. Hence $\lambda(S\setminus G)=0$, and $|\mu|(S\setminus G)=0$ because $|\mu|\ll\lambda$.

\textit{(ii)} By definition of the partition $\Q$, $Q(s)=Q(t)$ is equivalent to $T_U(s)=T_U(t)$, which
is what we show. Note first that at a type $z\in G\cap\bd S$ the outward normal is unique and
differentiability of $U$ along $\bd S$ fixes the tangential component of every subgradient, so
$\partial U(z)$ is a segment in the direction $\n_S(z)$; since $\hat{s}$ lies in the interior of
$S$, the selection $T_U(z)$ is the endpoint of that segment with the smallest component along
$\n_S(z)$.

If $U$ is differentiable at $t$, then $\partial U(t)=\{T_U(t)\}$ and the hypothesis
$T_U(s)\in\partial U(t)$ gives $T_U(s)=T_U(t)$. Otherwise $t\in G\cap\bd S$, and the preceding
observation gives
\[
T_U(s)=T_U(t)+r\,\n_S(t) \qquad\text{for some } r\ge0 .
\]
Suppose $r>0$. Since $T_U(s)$ is a subgradient at both $s$ and $t$, we have
$U(s)-U(t)=T_U(s)\cdot(s-t)$, while $T_U(t)\in\partial U(t)$ gives $U(s)-U(t)\ge T_U(t)\cdot(s-t)$.
Subtracting, $r\,\n_S(t)\cdot(s-t)\ge0$, whereas convexity of $S$ gives $\n_S(t)\cdot(s-t)\le0$.
Hence $\n_S(t)\cdot(s-t)=0$, so $s$ lies in the supporting hyperplane of $S$ at $t$. That
hyperplane therefore also supports $S$ at $s$, so $s\in\bd S$ and, since $s\in G$,
$\n_S(s)=\n_S(t)$. It also follows that $T_U(t)\cdot(s-t)=T_U(s)\cdot(s-t)$, so the plane with
slope $T_U(t)$ touches $U$ at $s$ as well, that is, $T_U(t)\in\partial U(s)$. But then
$T_U(t)=T_U(s)-r\,\n_S(s)$ is a subgradient at $s$ with a strictly smaller component along
$\n_S(s)$ than $T_U(s)$, contradicting the selection at $s$. Hence $r=0$.
\end{proof}

\begin{proof}[Proof of \autoref{thm:main_result}]
Throughout, $U$ is a convex function on $\R^n$ satisfying $U\le h$ (\autoref{lemma:feasible_set}); the compact domain $X$ of \autoref{lemma:compact_domain} is used only to construct the dual variable $\gamma$.

Let $\{\mu|_Q\}_{Q\in\Q}$ be regular, proper conditional measures of $\mu$ given the partition $\Q$ \parencite[][Example 10.4.11]{bogachev07b}. They are concentrated on $Q\cap S$ for $|\mu|$-almost every $Q$ and satisfy
\[
\mu=\int \mu|_{Q(s)}\,\mathrm d|\mu|(s).
\]
Since $\mathrm d\mu/\mathrm d|\mu|$ takes the values $\pm1$, the positive and negative parts of $\mu|_Q$ are the conditional measures of $\mu^+$ and $\mu^-$, with the same outer weighting.

\textit{Sufficiency.} Suppose the conditional inequalities hold, and let $V$ be any feasible indirect utility. For $|\mu|$-almost every $Q\in\Q$,
\[
\int V\,\mathrm d\mu|_Q
\le\int V\,\mathrm d\delta_Q
\le\int U\,\mathrm d\delta_Q
=\int U\,\mathrm d\mu|_Q.
\]
The first inequality is convex order, and the second uses $\delta_Q\ge0$ and its concentration on $\{s:U(s)=h(s)\}$. The equality follows because $U$ agrees with an affine function on $Q$ and convex order preserves mass and first moments. It also holds when $\delta_Q=0$. Integrating over partition elements with respect to $|\mu|$ proves optimality. All the integrals are finite: subgradients of $U$ on the compact set $S$ are bounded, so the types whose ideal actions equal these subgradients, including the relevant contact points, lie in a compact set.

\textit{Necessity.} By \autoref{l:strong_duality}, optimality gives a positive measure $\gamma$ satisfying \eqref{eq:cs1}, \eqref{eq:cs2}, and $\gamma\ge_{cx}\mu$. Strassen's theorem \parencite[see, for example,][pp.\ 93--94]{phelps01} supplies a probability kernel $D_s$ such that
\[
\gamma+\mu^-=\int D_s\,\mathrm d\mu^+(s),
\qquad
\int t\,\mathrm dD_s(t)=s
\quad\text{for $\mu^+$-almost every $s$}.
\]
Combining the first identity with \eqref{eq:cs2} gives
\[
\int\Big[\int U(t)\,\mathrm dD_s(t)-U(s)\Big]\,\mathrm d\mu^+(s)=0.
\]
Each bracket is nonnegative by Jensen's inequality, so it vanishes for $\mu^+$-almost every $s$. Let $p$ be the subgradient selected at such a point $s$. Since $D_s$ has barycenter $s$,
\[
\int[U(t)-U(s)-p\cdot(t-s)]\,\mathrm dD_s(t)=0.
\]
The integrand is nonnegative because $p\in\partial U(s)$, so it vanishes $D_s$-almost everywhere: the plane selected at $s$ touches $U$ at $D_s$-almost every $t$, that is, $p\in\partial U(t)$.

We show next that $D_s$ is concentrated on $Q(s)$ for $\mu^+$-almost every $s$. Since $\mu^+\le|\mu|$, \autoref{l:partition}(i) gives $s\in G$ for $\mu^+$-almost every $s$, and likewise $\mu^-(S\setminus G)=0$. By \eqref{eq:cs1}, $\gamma$ is concentrated on $\{t:U(t)=h(t)\}$. Since $\int D_s\,\mathrm d\mu^+(s)=\gamma+\mu^-$, it follows that for $\mu^+$-almost every $s$, $D_s$-almost every $t$ either lies in $G$ or satisfies $U(t)=h(t)$. Fix such an $s$ and $t$. If $t\in G$, then by the preceding paragraph the plane selected at $s$ touches $U$ at $t$, so \autoref{l:partition}(ii) gives $Q(t)=Q(s)$. If $U(t)=h(t)$, then $U\le h$ and differentiability of $h$ imply that $U$ has a unique supporting plane at $t$; by the preceding paragraph it is the plane selected at $s$, so again $Q(t)=Q(s)$.

For each partition element, define
\[
\alpha_Q:=\int D_s\,\mathrm d\mu|_Q^+(s)-\mu|_Q^-.
\]
Integrating over the partition and using the disintegrations of $\mu^+$ and $\mu^-$ recorded above,
\[
\int \alpha_{Q(s)}\,\mathrm d|\mu|(s)=\int D_t\,\mathrm d\mu^+(t)-\mu^-=\gamma .
\]
 Since $D_s$ is concentrated on $Q(s)$, each $\alpha_Q$ is concentrated on $Q$, so the $\alpha_Q$ are conditional measures of $\gamma$ given $\Q$, with the same outer weighting by $|\mu|$. Because $\gamma$ is positive, $\alpha_Q$ is positive for $|\mu|$-almost every $Q$, and by \eqref{eq:cs1} it is concentrated on $Q\cap\{U=h\}$.

Because $D_s$ has barycenter $s$, Jensen's inequality gives $\mu|_Q\le_{cx}\alpha_Q$. Moreover, $\alpha_Q(Q)=\mu|_Q(Q)$. There is at most one contact point in $Q$, so $\alpha_Q=\delta_Q$. If there is no contact point, $\alpha_Q=0$ and $\mu|_Q(Q)=0$. This proves the necessary conditions.
\end{proof}

\begin{proof}[Proof of \autoref{cor:onedimension}]
The partition consists of $(-\infty,s_1]$, $[s_2,\infty)$, and the singletons in $(s_1,s_2)$. On a singleton cell, $\mu|_Q$ is a point mass, so the convex-order comparison holds automatically and only nonnegative mass is required. This gives (i).

On the upper tail, convex order compares $\mu$ restricted to $[s_2,\overline{s}]$ with a point mass of the same total mass at $s_2$. Equality in (ii) matches their first moments, and the inequalities for $s>s_2$ give the remaining convex-order comparisons, since the functions $(x-s)^+$ with $s\le s_2$ are affine on $[s_2,\overline{s}]$. It remains to check that the tail mass is nonnegative. If $s_2<\overline{s}$, the expression in (ii) is zero at $s_2$ and nonpositive to its right, so its right derivative at $s_2$ satisfies
\[
-\mu([s_2,\overline{s}])\leq0,
\]
where we use that $\mu$ has no atom at the interior point $s_2$. If $s_2=\overline{s}$, the tail consists only of the endpoint mass, which is nonnegative by (i). The lower-tail comparison follows analogously from (iii). Thus the conditions are precisely those of \autoref{thm:main_result}.
\end{proof}

\begin{proof}[Proof of \autoref{cor:suff_one_dim}]
Enlarging $I$ by an adjacent interval on which $\mu$ vanishes preserves the hypothesis, so we may assume that $I$ admits no such enlargement. Let $U$ be an optimal indirect utility. By \autoref{thm:main_result}, every pool has nonnegative mass, so a pool that contains negative weight also contains positive weight and therefore meets $I$. Since pools are disjoint intervals, all negative weight below $I$ lies in a single pool $Q_1$, and all negative weight above $I$ in a single pool $Q_2$.

If some pool carries all virtual weight, assigning its lottery's mean to every type yields an optimal singleton delegation set. Otherwise, suppose there is negative weight below $I$. The weight of $Q_1$ is nonpositive below $I$ and nonnegative in $I$. Its total mass is positive: with zero mass, $\mu|_{Q_1}\le_{cx}\delta_{Q_1}$ would force a zero first moment, whereas the positive weight lies to the right of the negative weight. So $\delta_{Q_1}$ is a point mass at the signed barycenter $s_1$ of $\mu|_{Q_1}$, where $U$ is tangent to $h$. This barycenter lies to the right of the mean of the positive weight, so it is not below $I$. It is not above $I$ either: by maximality, either types arbitrarily close above $I$ carry negative weight and belong to $Q_2$, or $I$ extends to $\overline{s}$ and $Q_1$ would carry all virtual weight. If there is no negative weight below $I$, maximality implies that $I$ extends to $\underline{s}$, and we set $s_1=\underline{s}$. Define $s_2$ symmetrically.

The indirect utility of delegating to $[s_1,s_2]$ is the tangent to $h$ at $s_1$ below $s_1$, $h$ on $[s_1,s_2]$, and the tangent to $h$ at $s_2$ above $s_2$. Below $s_1$ it therefore agrees with $U$ on $Q_1$, which, when it exists, contains all weight below $s_1$, and similarly above $s_2$. On $[s_1,s_2]\subseteq I$ it equals $h\ge U$, and $\mu$ is nonnegative there. Hence it yields at least $\int U\,\mathrm d\mu$, and delegating to $[s_1,s_2]$ is optimal.
\end{proof}

\begin{proof}[Proof of \autoref{cor:log_concave}]
By \eqref{eq:def_nu}, $\nu/f=1-\beta(\log f)'$ in the interior of the support. If $\beta>0$, this expression is nondecreasing, while the endpoint masses are $-\beta f(\underline{s})\leq0$ and $\beta f(\overline{s})\geq0$. Thus $\mu$ is nonnegative on an upper interval and nonpositive below it, as required by \autoref{cor:suff_one_dim}. For $\beta<0$ the signs and direction reverse. If $\beta=0$, then $\mu=F$ and full discretion is optimal.
\end{proof}

\begin{proof}[Proof of \autoref{cor:convex_delegation}]
Since $b(a)=-\norm{a}^2/2$, each type chooses his Euclidean projection onto $A$. This choice is unique, so the indirect utility is differentiable and its partition consists of the singletons in $\interior A$ and the outward normal rays $Q_a=a+N_A(a)$ for $a\in\bd A$. Each ray meets $h$ only at $a$.

On a singleton cell, $\mu|_Q$ is a point mass, so the convex-order comparison holds and \autoref{thm:main_result} only requires nonnegative mass, which is equivalent to (i).

Fix a normal ray $Q_a$ and write its points as $x=a+r\,\n_A(a)$ with $r\ge0$, where $\n_A(a)$ is the outward unit normal at $a$. The theorem compares $\mu|_{Q_a}$ with the point mass $\delta_{Q_a}$ of mass $\mu|_{Q_a}(Q_a)$ at $a$. These measures are one-dimensional, and by the characterization recalled before \autoref{cor:onedimension} it amounts to equal masses, equal first moments, and the comparisons for the functions $(r-t)^+$.

For $s=a+t\,\n_A(a)$ with $t\ge0$, we have $x\in s+N_A(a)$ exactly when $r\ge t$, and then $\norm{x-s}=r-t$. The integral in (ii) is therefore $\int(r-t)^+\,\mathrm d\mu|_{Q_a}$. At $t=0$ it is the first moment of $\mu|_{Q_a}$ about $a$, so the equality in (ii) matches first moments. For $t>0$, the inequalities in (ii) are the comparisons with $\delta_{Q_a}$, which gives $(r-t)^+$ integral zero. For $t<0$, $(r-t)^+$ is affine on the ray, so matching masses and first moments suffice. Hence (ii) is equivalent to $\mu|_{Q_a}(Q_a)\ge0$ and $\mu|_{Q_a}\le_{cx}\delta_{Q_a}$. Together with (i), these are precisely the conditions of \autoref{thm:main_result}.
\end{proof}

\begin{proof}[Proof of \autoref{cor:flat_conditional}]
The budget leaves the agent unrestricted orthogonally to
$\bbeta$ and caps his choice in the remaining direction.
Types therefore separate below the cap and pool along rays
parallel to $\bbeta$ above it. 
Conditional on $Y$, the agent's ideal action $s$ has $\bbeta\cdot s=Z$ and the principal's ideal action $s-\bbeta$ has $\bbeta\cdot(s-\bbeta)=Z-\norm{\bbeta}^2$. Since the bias has no orthogonal component,
\eqref{eq:reformulate_U} gives exactly the virtual-value
comparisons of this one-dimensional problem.
Thus the singleton and pooling comparisons in
\autoref{thm:main_result} coincide with those for the conditional
cap, proving the claim.
\end{proof}

\begin{proof}[Proof of \autoref{cor:symmetric_logconcave}]
Let $f_0$ be the common density. Conditional on $s_1-s_2=d$, the density of $s_1+s_2$ at $w$ is proportional to $f_0((w+d)/2)f_0((w-d)/2)$, so it is log-concave and symmetric about $2m$. Conditioning on $s_1-s_2$ amounts to conditioning on $Y$, so $Z=\beta(s_1+s_2)$ is conditionally log-concave with $\E[Z\mid Y]=2\beta m$.

By \autoref{cor:flat_conditional}, it suffices to show that, for almost every $Y$, the cap $c=2\beta(m-\beta)$ is optimal in the one-dimensional problem with state $Z$ and principal's ideal action $Z-\norm{\bbeta}^2$, where $\norm{\bbeta}^2=2\beta^2$. By \eqref{eq:reformulate_U}, the virtual measure of this problem integrates a function $V$ to $\E[V(Z)\mid Y]-\norm{\bbeta}^2\,\E[V'(Z)\mid Y]$. It therefore has mass one and first moment $\E[Z\mid Y]-\norm{\bbeta}^2=c$, and it integrates $(Z-t)^+$ to
\[
\E[(Z-t)^+\mid Y]-\norm{\bbeta}^2\,\Prob(Z>t\mid Y).
\]
The bias bound gives $m-\beta\le\underline{s}$, so $c\le2\beta\underline{s}\le Z$ almost surely. For $t<c$, the function $(Z-t)^+$ is affine on the support, so its integral is $c-t$, as for a unit point mass at $c$. For $t\ge c$ with $\Prob(Z>t\mid Y)>0$, the integral equals $\Prob(Z>t\mid Y)\bigl(\E[Z-t\mid Z>t,Y]-\norm{\bbeta}^2\bigr)$. The mean residual life $\E[Z-t\mid Z>t,Y]$ of a log-concave distribution is nonincreasing in $t$, and at $t=c$, where $c\le Z$, it equals $\E[Z\mid Y]-c=\norm{\bbeta}^2$. Hence the integral is nonpositive for $t\ge c$, where a unit point mass at $c$ gives zero. The virtual measure is thus smaller in convex order than a unit point mass at $c$, and \autoref{thm:main_result} makes pooling at $c$ optimal conditional on almost every $Y$. So the budget $a_1+a_2\le2(m-\beta)$ is optimal and binds almost surely.

It remains to show that, with a single decision, complete pooling is optimal exactly when $\beta\ge(\overline{s}-\underline{s})/2$. By the first-moment condition, the pooled action must be the principal's ex ante ideal action $m-\beta$. If $m-\beta\le\underline{s}$, the argument above, with $Z$ replaced by $s$ and $\norm{\bbeta}^2$ by $\beta$, proves optimality. Otherwise $\underline{s}<m-\beta<m$, so $\E[s\mid s>m-\beta]>m$, and the comparison for $(s-(m-\beta))^+$ fails.
\end{proof}

\begin{proof}[Proof of \autoref{prop:negative_boundary_hull}]
Fix an optimal $U$, and let $\overline{U}$ be the convex envelope of the function $\psi$ equal to $U$ on $\supp\mu^-\subseteq\bd S$ and to $h$ elsewhere. Then $U\leq \overline{U}\leq h$, $\overline{U}=U$ on $\supp\mu^-$, and $\overline{U}$ is feasible by \autoref{lemma:feasible_set}. Since $\overline{U}-U\geq0$ vanishes on the negative support, $\int(\overline{U}-U)\,\mathrm d\mu\geq0$. Optimality, strict positivity of $\nu$ almost everywhere in $\interior S$, and continuity therefore imply $\overline{U}=U$ throughout $\interior S$.

Since $\supp\mu^-$ is compact and $U\leq h$, $\psi$ is lower semicontinuous. Because $b$ is finite on $\R^n$, $h=(-b)^*$ and hence $\psi$ are superlinear. Its convex envelope therefore admits an attained Carath\'eodory representation. Fix $x\in\interior S\setminus K$ and a minimizing convex combination with positive weights,
\[
x=\sum_i\theta_i z_i,\qquad \overline{U}(x)=\sum_i\theta_i\psi(z_i).
\]
Some $z_j$ lies outside $\supp\mu^-$, since $x\notin K$. For any affine support $p$ of $\overline{U}$ at $x$, $p\leq \overline{U}\leq\psi$ and equality in this combination force $p(z_i)=\psi(z_i)$ for each $i$. In particular, $p(z_j)=h(z_j)$ and $p\leq h$, so $p$ is the unique tangent to $h$ at $z_j$. Hence $\overline{U}$, and thus $U$, is differentiable at $x$, with gradient $a=\nabla h(z_j)$, the ideal action of type $z_j$. Any mechanism inducing $U$ assigns $x$ a lottery with mean $a$ and expected value of $b$ equal to $U(x)-a\cdot x=p(x)-a\cdot x=b(a)$, so strict concavity of $b$ forces the lottery to be degenerate.
\end{proof}

\begin{proof}[Proof of \autoref{cor:negative_face_connected}]
The convex hull of the negative support stays on that face and therefore does not meet the interior of $S$, so \autoref{prop:negative_boundary_hull} applies throughout the interior. As the type varies over this connected set, the chosen action varies continuously, so the set of actions used by interior types is connected, and its closure is compact and connected as well. Offering exactly these actions reproduces the optimal utility for interior types, and by continuity also for boundary types. Relabeling them by their ideal types, as in our definition of delegation, preserves compactness and connectedness.
\end{proof}

\clearpage
\section{Supplementary appendix}\label{app:online}

\subsection{Proofs for Section~\ref{sec:general_payoffs}}

\begin{proof}[Proof of \autoref{thm:main_result_nonlinear}]
Let $\tilde{m}$ be an incentive-compatible mechanism with indirect utility $\tU$. At almost every $s$, the expected action under $\tilde{m}(s)$ is $\nabla \tU(s)$ and $b(\tilde{m}(s))=\tU(s)-s\cdot\nabla \tU(s)$; since $m$ is deterministic it chooses the action $\nabla U(s)$ and satisfies the same identity. Applying the tangent inequality for the concave function $a\mapsto u_P(a,s)-\kappa b(a)$ at $\nabla U(s)$, averaging over $\tilde{m}(s)$, and substituting these identities gives
\begin{align*}
u_P(\tilde{m}(s),s)-u_P(m(s),s) \le{}& \kappa\big[\tU(s)-U(s)\big] \\
&+ \big[\nabla_a u_P(\nabla U(s),s) -\kappa s - \kappa \nabla b(\nabla U(s))\big]\cdot \big[\nabla \tU(s)-\nabla U(s)\big].
\end{align*}
Integrating and applying \eqref{eq:def_tilde_mu} to $\tU-U$ yields
\[ \int \big[u_P(\tilde{m}(s),s)-u_P(m(s),s)\big] \,\mathrm dF(s) \le \int (\tU-U) \,\mathrm d\tilde{\mu}. \]

It remains to show that the right-hand side is not positive. For a.e.\ $Q\in\Q$,
\[ \int \tU \,\mathrm d\tilde{\mu}|_Q \le \int \tU \,\mathrm d\delta_Q \le \int U \,\mathrm d\delta_Q = \int U \,\mathrm d\tilde{\mu}|_Q, \]
where the first inequality holds because $\tU$ is convex and $\tilde{\mu}|_Q\le_{cx}\delta_Q$, the second because $\delta_Q$ is nonnegative and concentrated where $U=h$ while $\tU\le h$, and the equality because $U$ is affine on $Q$ and convex order preserves the integrals of affine functions; the case $\delta_Q=0$ is covered as well. Integrating over the partition elements gives $\int (\tU-U) \,\mathrm d\tilde{\mu}\le 0$, so no incentive-compatible mechanism yields the principal a higher payoff than $m$.
\end{proof}

\begin{proof}[Proof of \autoref{prop:weighted_quadratic}]
Let $m$ be a mechanism inducing $U$. At almost every type, $\E[m(s)]=\nabla U(s)$ by the envelope theorem and $U(s)=\nabla U(s)\cdot s-\frac{1}{2}\E\norm{m(s)}^2$, so the covariance matrix $\Sigma(s)$ of $m(s)$ satisfies
\[ \operatorname{tr}\Sigma(s)=\E\norm{m(s)}^2-\norm{\nabla U(s)}^2=\sigma^2(s). \]
The principal's expected payoff at $s$ is $u_P(\nabla U(s),s)-\frac{1}{2}\operatorname{tr}\big(W\Sigma(s)\big)$, and since $W-\kappa I$ and $\Sigma(s)$ are positive semidefinite,
\[ \operatorname{tr}\big(W\Sigma(s)\big)=\operatorname{tr}\big((W-\kappa I)\Sigma(s)\big)+\kappa\operatorname{tr}\Sigma(s)\ge \kappa\sigma^2(s). \]
That payoff is therefore at most $u_P(\nabla U(s),s)-\kappa\sigma^2(s)/2$, the integrand of \eqref{eq:general_payoff_bound}. The lottery in the statement has mean $\nabla U(s)$ and places all of its variance in the direction $e$, so $\operatorname{tr}\Sigma(s)=\sigma^2(s)$ and $\operatorname{tr}(W\Sigma(s))=\sigma^2(s)\,e\cdot We=\kappa\sigma^2(s)$, and it attains equality. Replacing each assigned lottery in this way leaves its mean and expected squared norm unchanged, hence every type's payoff from every report; the modified mechanism is therefore incentive compatible and still induces $U$. Integrating gives $\Phi(U)$.

Expanding $\Phi$ gives \eqref{eq:weighted_quadratic_objective}, which is quadratic in $U$, so for every feasible $V$
\[ \Phi(V)-\Phi(U)=\int (V-U) \,\mathrm d\tilde{\mu}-\frac{1}{2}\int \nabla(V-U)\cdot (W-\kappa I)\nabla(V-U) \,\mathrm dF . \]
The last term is nonpositive, so $\int (V-U) \,\mathrm d\tilde{\mu}\le 0$ for every feasible $V$ suffices for optimality. It is also necessary: $U+\varepsilon(V-U)$ is feasible for $\varepsilon\in[0,1]$, and the derivative of $\Phi$ along this segment at $\varepsilon=0$ equals $\int (V-U) \,\mathrm d\tilde{\mu}$.
\end{proof}

\begin{proof}[Proof of \autoref{prop:aligned_randomization}]
Write $a=(a^x,a^z)$ and measure the principal's payoff by her expected
loss relative to full information:
\[
 L_\kappa(m):=\frac12\E\left[
 \norm{a^x-(x-\bbeta)}^2+\kappa\norm{a^z-z}^2\right],
 \qquad a\sim m(x,z).
\]
The expectation includes the state and the mechanism's randomization.

\emph{An inexpensive stochastic implementation.}
Let $X$ be the projection of $S$ onto the biased coordinates,
choose $x_0\in\arg\min_{x\in X}\bbeta\cdot x$, and let $e$ be a
coordinate unit vector in the aligned block. Define $m_0$ to
assign each of the actions
\[
 \left(x-\bbeta,\;
 z\pm\sqrt{2\bbeta\cdot(x-x_0)}\,e\right)
\]
with probability $1/2$. Its variance is nonnegative by the choice
of $x_0$. On $S$, its truthful utility is
\[
 U_0(x,z)=\frac12\big(\norm{x}^2+\norm{z}^2\big)
          -\frac12\norm{\bbeta}^2-\bbeta\cdot(x-x_0).
\]
For any $s,s'\in S$, direct substitution gives
\[
 u_A(m_0(s),s)-u_A(m_0(s'),s)=\frac12\norm{s-s'}^2.
\]
Thus $m_0$ is incentive compatible. Its global indirect utility is the
supremum of the offered lotteries' payoff planes; the displayed
formula for $U_0$ is asserted only on $S$. Setting
$C:=\E[\bbeta\cdot(x-x_0)]>0$, we obtain
\begin{align}\label{eq:aligned_loss_bound}
 \inf_m L_\kappa(m)\leq L_\kappa(m_0)=\kappa C.
\end{align}
No independence assumption is used here.

\emph{Deterministic mechanisms cannot match this order of loss.}
We show the stronger statement
\begin{align}\label{eq:aligned_deterministic_loss}
 \lim_{\kappa\downarrow0}
 \frac{\inf_{m\text{ deterministic and IC}}L_\kappa(m)}{\kappa}
 =+\infty.
\end{align}
Otherwise there would be $\kappa_j\downarrow0$, a finite $M$, and
deterministic incentive-compatible action rules $a_j$ satisfying
$L_{\kappa_j}(a_j)\leq M\kappa_j$. Consequently,
\[
 \E\norm{a_j^x-(x-\bbeta)}^2\longrightarrow0,
 \qquad \E\norm{a_j^z-z}^2\leq2M,
\]
so the second moments of $a_j$ are uniformly bounded. Let $U_j$ be
their global indirect utilities. On $S$,
$U_j(s)=s\cdot a_j(s)-\norm{a_j(s)}^2/2$, so
$\E[U_j(s)]$ is bounded below uniformly in $j$.
For any $r\in\R^n$, convexity and $U_j\leq h$ give
\[
 U_j(s)\leq\frac12 U_j(r)+\frac12 h(2s-r).
\]
Integrating over $s$ bounds $U_j(r)$ below uniformly on each compact
set. Together with $U_j\leq h$, these bounds imply local
equicontinuity. Passing to a subsequence, $U_j$ converges locally
uniformly to a finite convex function $U\leq h$, and
$\nabla U_j\to\nabla U$ at almost every point where $U$ is
differentiable. Since $F$ has a density, we obtain, $F$-almost
everywhere,
\[
 \nabla_x U(x,z)=x-\bbeta,
 \qquad
 U(s)=s\cdot\nabla U(s)-\frac12\norm{\nabla U(s)}^2.
\]
Choose an open product box inside $S$ on which the density is
positive; such a box exists by continuity of the density. The first
identity implies there that
$U(x,z)=\norm{x}^2/2-\bbeta\cdot x+w(z)$. Substituting into the
second identity yields
\[
 -\bbeta\cdot x+w(z)
 =-\frac12\norm{\bbeta}^2
   +z\cdot\nabla w(z)-\frac12\norm{\nabla w(z)}^2
\]
almost everywhere on that box, which is impossible for
$\bbeta\neq0$. This proves \eqref{eq:aligned_deterministic_loss}.
Together with \eqref{eq:aligned_loss_bound}, it gives a strict
advantage over all deterministic mechanisms for sufficiently small
$\kappa$, without independence.

\emph{Optimal implementation under independence.}
Suppose now that $x$ and $z$ are independent. For any feasible $U$,
define, on the whole space,
\[
 v(x):=\E_z\left[U(x,z)-\frac12\norm{z}^2\right],
 \qquad
 \bar U(x,z):=v(x)+\frac12\norm{z}^2.
\]
The function $v$ is convex and satisfies $v(x)\leq\norm{x}^2/2$,
so $\bar U$ is feasible. In
\eqref{eq:weighted_quadratic_objective}, the coefficient of
$\nabla_z U$ is zero and the quadratic term is
$-(1-\kappa)\norm{\nabla_x U}^2/2$. Independence therefore gives
\[
 \Phi(\bar U)-\Phi(U)
 =\frac{1-\kappa}{2}\E\left[
   \norm{\nabla_x U(x,z)-\E_z[\nabla_x U(x,z)]}^2\right]\geq0.
\]
An optimum can consequently be chosen with indirect utility
$v_\kappa(x)+\norm{z}^2/2$. Applying
\autoref{prop:weighted_quadratic} gives
\[
 q_\kappa(x)=\nabla v_\kappa(x),\qquad
 \sigma_\kappa(x)^2
 =2\left[x\cdot q_\kappa(x)
          -\frac12\norm{q_\kappa(x)}^2-v_\kappa(x)\right].
\]
The resulting mechanism assigns
$\big(q_\kappa(x),z\pm\sigma_\kappa(x)e\big)$ with equal
probabilities, so every aligned action has conditional mean
equal to its common ideal. Without independence, the conclusion
that all randomization can be concentrated in a single aligned
decision follows directly from
\autoref{prop:weighted_quadratic}.

For completeness, an optimum exists for each fixed $\kappa>0$:
bounded loss bounds second moments, so the compactness argument
above applies also to stochastic mechanisms. After choosing the
least-cost lotteries of \autoref{prop:weighted_quadratic}, their
means and total variances converge almost everywhere along the
resulting subsequence. Fatou's lemma for the nonnegative quadratic
loss then gives a feasible minimizer.
\end{proof}

\subsection{Proofs for Section~\ref{sec:outside_option}}

\begin{proof}[Proof of \autoref{l:strong_duality_outside}]
We parametrize the primal problem: for $g\in C(X)$, let $V(g)$ denote the value of the primal problem with upper bound $g$ and let
\[\mathcal{F}(g) := \{U\in\U_0: U\le g\}\]
denote its feasible set. The value $V(h)$ is finite because $0\le U\le h$ on $S$ for every $U\in\mathcal{F}(h)$. It is enough to show that $\frac{V(g)-V(h)}{\norm{g-h}}$ is bounded above \citep[see, e.g., Condition 3, p. 266, in][]{gretsky-etal2002}, since this yields a subgradient of $-V$ at $h$ and hence the result \citep[][Proposition 2.5]{shapiro10}.

Given $\tilde{U}\in \mathcal{F}(g)$, define $U:= \max\{\tilde{U}-\norm{g-h},0\}$. This function is convex and nonnegative, and it satisfies $U\le h$ on $X$ because $\tilde{U}\le g\le h+\norm{g-h}$ and $h\ge 0$; hence $U\in\mathcal{F}(h)$. On $S$ we have $\tilde{U}\ge 0$ and therefore $0\le \tilde{U}-U\le \norm{g-h}$. Since $\mu$ is concentrated on $S$, it follows that
\[ \int \tilde{U} \,\mathrm d\mu - \int U \,\mathrm d\mu \le \norm{\mu}_{TV}\norm{g-h}. \]
Taking the supremum over $\tilde{U}\in\mathcal{F}(g)$ bounds $\frac{V(g)-V(h)}{\norm{g-h}}$ by $\norm{\mu}_{TV}$, which gives strong duality and dual attainment.
\end{proof}

\begin{proof}[Proof of \autoref{thm:outside_option}]
\textit{Sufficiency.} Let $V$ be any feasible indirect utility. For $|\mu|$-a.e.\ $Q\in\Q$,
\[ \int V \,\mathrm d\mu|_Q \le \int V \,\mathrm d\delta_Q \le \int U \,\mathrm d\delta_Q = \int U \,\mathrm d\mu|_Q, \]
where the first inequality holds because $V\in\U_0$, the second because $\delta_Q$ is nonnegative, $V\le h$, and $U=h$ on the support of $\delta_Q$, and the equality is the second condition of the theorem. Integrating over the partition elements shows that $U$ is optimal.

\textit{Necessity.} Suppose $U$ is optimal. By \autoref{l:strong_duality_outside} there is a positive measure $\gamma$ with $\mu\le_{\U_0}\gamma$ such that $U=h$ $\gamma$-almost everywhere and $\int U \,\mathrm d\mu = \int U \,\mathrm d\gamma$.

The order $\le_{\U_0}$ can be converted into the convex order at the cost of an additional measure. Consider maximizing $\int V \,\mathrm d(\mu-\gamma)$ over convex $V$ that are nonnegative on $S$. Its value is zero because $\mu\le_{\U_0}\gamma$, and the constant function $1$ satisfies the constraint strictly, so there is a positive measure $\beta$ on $S$ with $\mu+\beta\le_{cx}\gamma$ \parencite[][Proposition 2.8]{shapiro10}. Since $\int U \,\mathrm d\mu = \int U \,\mathrm d\gamma$ and $U\ge 0$ on $S$, this forces $\int U \,\mathrm d\beta = 0$, so $\beta$ is concentrated on the types with $U(s)=0$.

By Strassen's theorem \parencite[see, for example,][pp.\ 93--94]{phelps01}, $\gamma+\mu^-$ is a mean-preserving spread of $\mu^++\beta$, and as in the proof of \autoref{thm:main_result} the equality of the integrals of $U$ makes Jensen's inequality tight at almost every source.

We may take $\beta$ to be a part of $\mu^-$. Indeed, every source in $\beta$ has $U(s)=0$, every target satisfies $U\ge 0$ (targets lie in $S$, or in the support of $\gamma$, where $U=h\ge 0$), and Jensen's inequality is tight, so all mass originating in $\beta$ lands where $U=0$. Delete this part of the spread together with its targets, split the latter into their $\gamma$ and $\mu^-$ parts, and reinstate the deleted $\mu^-$ part by leaving that mass in place as a source. Writing $\gamma$ and $\beta$ for the measures that result, the spread now runs from $\mu^++\beta$ to $\gamma+\mu^-$, with $\beta\le\mu^-$ concentrated on $\{s:U(s)=0\}$. Since the deleted part of $\gamma$ also sits where $U=0$, the equality $\int U \,\mathrm d\mu = \int U \,\mathrm d\gamma$ is preserved.

All sources are now dominated by $|\mu|$, so the argument in the proof of \autoref{thm:main_result} applies unchanged: the spread does not move mass across partition elements, even where $U$ has boundary kinks. Conditioning as in the proof of \autoref{thm:main_result} and letting $\delta_Q$ denote the conditional measure of $\gamma$, which is nonnegative and concentrated on $Q\cap\{s:U(s)=h(s)\}$, we obtain
\[ \mu|_Q+\beta|_Q \le_{cx} \delta_Q \]
for almost every $Q$. For $V\in\U_0$ the term $\int V \,\mathrm d\beta|_Q$ is nonnegative and can be dropped, which gives $\mu|_Q\le_{\U_0}\delta_Q$; for $V=U$ it vanishes, which gives the asserted equality.
\end{proof}

\subsection{Monopoly regulation}\label{app:monopoly}

This subsection gives the assumptions and formal results for the monopoly regulation application in \autoref{sec:outside_option}. A monopolist is privately informed about his one-dimensional productivity parameter $s\in S=[\underline{s},\overline{s}]$, which determines the constant marginal cost of production. A monopolist with type $s$ who produces output $a>0$ obtains profits $s a +b(a)-c$, where $c\ge 0$ is a fixed cost of production. Alternatively, the monopolist can shut down and produce nothing, which yields profits $0$.

The regulator's payoff is a weighted sum of producer and consumer surplus,
\[ u_P(a,s):= s a +b(a) - c\mathbf{1}_{a>0} + v(a), \qquad a\ge 0, \]
where $v(a)$ is the weighted consumer surplus resulting from quantity $a$. We assume $b$ is strictly concave and continuously differentiable on $[0,\infty)$, satisfies $b(0)=0$, and has $\lim_{a\rightarrow\infty} b(a)/a=-\infty$. We assume $v$ is continuous on $[0,\infty)$ and continuously differentiable on $(0,\infty)$, satisfies $v(0)=0$, and that $b+v$ is concave. The regulator's payoff is therefore concave in the output produced, but jumps at shutdown whenever $c>0$.

The regulator chooses an incentive-compatible mechanism to maximize her expected payoff. As in the general model, mechanisms may assign lotteries, and feasibility requires that every type obtain a nonnegative expected profit. Writing $u_A(a,s)= s a + b(a) - c\mathbf{1}_{a>0}$ and including the shutdown option when extending the indirect utility to fictitious types,
\[ U(s):= \max\Big\{0,\ \sup_{s'\in S} u_A(m(s'),s)\Big\}, \qquad h(s):= \sup_{a\ge 0} \big\{ s a + b(a) - c\mathbf{1}_{a>0}\big\}. \]
This extension agrees with the monopolist's expected profit on $S$; it is convex and satisfies $0\le U\le h$ on $\R$. Let $s_0:=\inf\{s: h(s)>0\}$ denote the lowest type for whom production is profitable if no quantity constraints are imposed.

\begin{lemma}\label{l:feasible_set_monopolist}
A function on $S$ is induced by a feasible mechanism if and only if it has a convex extension $U$ to $\R$ satisfying $0\le U\le h$.
\end{lemma}

The proof follows the supporting-line argument in \autoref{lemma:feasible_set}. The fixed cost changes which lotteries implement a given supporting line: a lottery between shutdown and a positive output can be needed even when the expected output is small.

We illustrate the approach by giving sufficient conditions for a \emph{cap allocation} \citep[see][]{AB:21} to be optimal. Given $\underline{a}\ge \lim_{s\downarrow s_0} h'(s)$, the cap allocation $q_{\underline{a}}$ lets the monopolist choose from $\{0\}\cup[\underline{a},\infty)$: types that strictly prefer shutdown produce nothing, and the remaining types choose their preferred output above the cap. This output floor corresponds to a price cap. Its indirect utility is illustrated in \autoref{fig:cap_allocation_2}.

\begin{figure}[ht]
\centering
\includegraphics[trim={40cm 0 5cm 0},clip,width=0.5\textwidth]{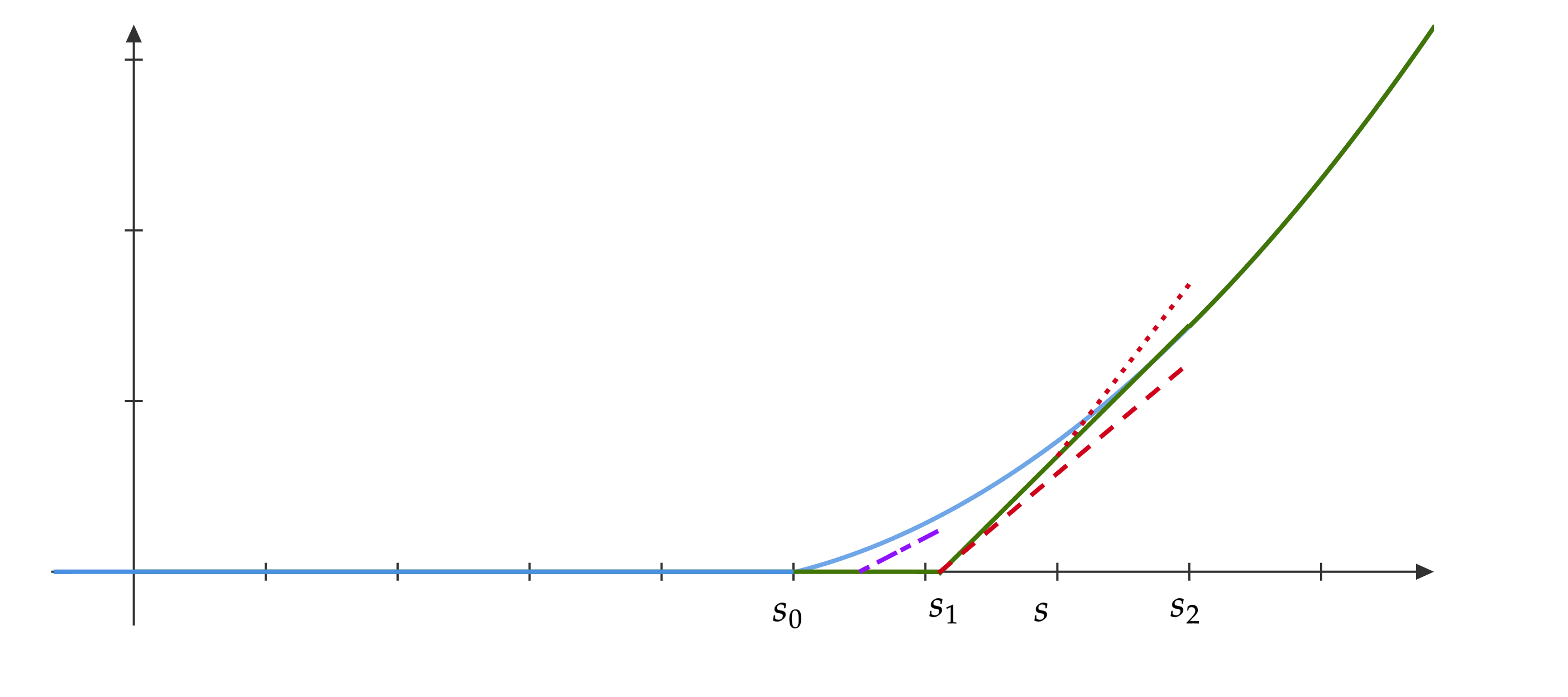}
\caption{Perturbations of a cap allocation. The blue curve is $h$ and the green curve the indirect utility of the cap allocation; the red and purple curves illustrate perturbations on the pooling and shutdown regions, respectively.}
\label{fig:cap_allocation_2}
\end{figure}

Fix a cap allocation with $\underline{a}> \lim_{s\downarrow s_0} h'(s)$ and indirect utility $U$, and let
\[ s_1:=\frac{c-b(\underline{a})}{\underline{a}}, \qquad s_2:=-b'(\underline{a}). \]
Then $s_0<s_1<s_2$: the monopolist shuts down below $s_1$, produces $\underline{a}$ between $s_1$ and $s_2$, and chooses his unconstrained output above $s_2$. These thresholds need not lie in $S$; in particular, $s_2$ remains the contact point with $h$ even when all producing types in $S$ choose the cap.

Because the regulator's payoff jumps at shutdown, the linearization of \autoref{sec:general_payoffs} must be replaced by a bound that is tight at shutdown as well. For this application, set
\[ \hat{\kappa}:= 1 + \frac{v(\underline{a})-\underline{a}v'(\underline{a})}{b(\underline{a})-\underline{a}b'(\underline{a})-c} \]
and assume that $a\mapsto (1-\hat{\kappa})b(a)+v(a)$ is concave on $(0,\infty)$. The denominator equals $\underline{a}(s_2-s_1)$ and is therefore positive. This coefficient is specific to the cap and need not agree with the coefficient $\kappa$ of \autoref{sec:general_payoffs}; it may be negative. Its defining property is that
\[ (1-\hat{\kappa})[b(\underline{a})-c]+v(\underline{a}) = \underline{a}\big[(1-\hat{\kappa})b'(\underline{a})+v'(\underline{a})\big], \]
so that the tangent at $\underline{a}$ to the adjusted payoff $a\mapsto (1-\hat{\kappa})[b(a)-c]+v(a)$ passes through the shutdown payoff at $a=0$. Following this tangent below $\underline{a}$ and the adjusted payoff above it gives a concave upper bound that is tight at every output the cap uses.

The measure of \autoref{sec:general_payoffs} takes the following form here, with $\hat{\kappa}$ in place of $\kappa$ and the concave bound in place of the adjusted payoff: $\tilde{\mu}$ is the signed measure supported on $S$ with
\begin{align}\label{eq:monopoly_measure}
\int V \,\mathrm d\tilde{\mu} = {}& \int_S \hat{\kappa}V(s) \,\mathrm dF(s) \\
& + \int_S \Big[(1-\hat{\kappa})\big(s+b'(\max\{\underline{a},U'(s)\})\big) + v'(\max\{\underline{a},U'(s)\})\Big]V'(s) \,\mathrm dF(s) \nonumber
\end{align}
for every continuously differentiable $V$ on a neighborhood of $S$.\footnote{The measure is well defined: before multiplying by $f$, the coefficient of $V'$ is the sum of a linear function of $s$ and the derivative of a concave function evaluated at a nondecreasing output, and therefore has bounded variation on $S$. Integrating by parts gives a finite signed measure, including its endpoint masses, and smooth approximation extends the identity to differences of feasible indirect utilities.} The coefficient of $V'$ is the slope of the concave bound at the output the cap assigns: it equals $(1-\hat{\kappa})(s+b'(\underline{a}))+v'(\underline{a})$ below $s_2$ and $v'(h'(s))$ above $s_2$, and the two agree at $s_2$.

\begin{proposition}\label{prop:optimality_cap_allocation}
Under the curvature condition above, the cap allocation is optimal if
\begin{enumerate}[(i)]
\item $\tilde{\mu}$ is nonnegative on $[s_2,\infty)$,
\item $\int_{(s_1,s)}(s-x) \,\mathrm d\tilde{\mu}(x)\le 0$ for all $s\in[s_1,s_2]$ with equality for $s=s_2$, and
\item $\int_{[s,s_1]}(x-s) \,\mathrm d\tilde{\mu}(x)\le 0$ for all $s\in[s_0,s_1]$.
\end{enumerate}
\end{proposition}

The interpretation is as before. Condition (i) ensures that the regulator benefits from giving types above $s_2$ their full-discretion payoff. Condition (ii) is the comparison on the pooling interval, illustrated by the red curve in \autoref{fig:cap_allocation_2}, and the equality at $s_2$ ensures that the affine piece of $U$ has the same value as the corresponding mass at its contact point. Condition (iii) rules out gains from raising the utility of the shutdown types, as illustrated by the purple curve. The comparisons use the entire measure, so any mass at an endpoint of $S$ is included; a mass at $s_1$ is assigned to the shutdown region and a mass at $s_2$ to the region above it.

\paragraph{The cap's first-order condition.}
To relate condition (ii) to \eqref{eq:monopoly_cap_balance}, suppose $\underline{s}<s_1<s_2<\overline{s}$ and write $J(\underline{a}):=\int_S u_P(q_{\underline{a}}(s),s)\,\mathrm dF(s)$. The shutdown threshold satisfies $\mathrm ds_1/\mathrm d\underline{a}=(s_2-s_1)/\underline{a}$. At $s_1$, operating rather than shutting down yields the regulator $v(\underline{a})$; at $s_2$, constrained and unconstrained output coincide. Differentiating the cap's expected payoff therefore gives
\[
 J'(\underline{a})
 =\int_{s_1}^{s_2}[s-s_2+v'(\underline{a})]f(s)\,\mathrm ds
  -\frac{v(\underline{a})}{\underline{a}}(s_2-s_1)f(s_1).
\]
For $s\leq s_2$, let $\ell(s):=(1-\hat{\kappa})(s-s_2)+v'(\underline{a})$ denote the coefficient of $V'$ in \eqref{eq:monopoly_measure}. The definition of $\hat{\kappa}$ gives $\ell(s_1)=v(\underline{a})/\underline{a}$. Integration by parts then yields
\begin{align*}
 \int_{(s_1,s_2)}(s_2-s)\,\mathrm d\tilde{\mu}(s)
 &=\int_{s_1}^{s_2}[s_2-s-v'(\underline{a})]f(s)\,\mathrm ds
   +(s_2-s_1)\ell(s_1)f(s_1)\\
 &=-J'(\underline{a}).
\end{align*}
Thus the equality in condition (ii) is exactly the cap's first-order condition. The proposition also applies when a cutoff lies outside $S$, without imposing the interior-cutoff restriction used for this interpretation.

\begin{proof}[Proof of \autoref{l:feasible_set_monopolist}]
For a feasible mechanism, the supremum of the payoff lines of its lotteries and of the shutdown option gives a convex extension satisfying $0\le U\le h$.

Conversely, let $U$ be such an extension and choose a supporting slope $a\in\partial U(s)$ for each $s\in S$. Since $h=0$ on $(-\infty,s_0]$, we have $U=0$ there, so every supporting slope is nonnegative. If $a=0$, then $U(s)=0$, and we assign shutdown.

Suppose $a>0$. The largest expected value of $b(a)-c\mathbf{1}_{a>0}$ at a fixed mean is given by its concave hull, which for $c>0$ joins the origin tangentially to $b-c$ and then follows $b-c$, and which for $c=0$ is simply $b$. Since the supporting line lies below $U\le h$, its intercept $U(s)-sa$ is at most this hull evaluated at $a$. We can therefore choose $A\ge a$ with
\[ \frac{b(A)-c}{A} = \frac{U(s)-sa}{a}. \]
Indeed, the largest value of the left-hand side over $A\ge a$, multiplied by $a$, is the hull evaluated at $a$, and the ratio decreases to $-\infty$ thereafter.

Assign output $A$ with probability $a/A$ and shutdown otherwise. This lottery has expected output $a$ and yields type $s$ the profit $U(s)$. Its payoff to any other type $t$ is $U(s)+a(t-s)$, a supporting line below $U(t)$, so the mechanism is incentive compatible, and these choices can be made measurably. Finally, $U(s)\ge 0$ bounds the target ratio below by $-s$, so the chosen outputs can be bounded uniformly on $S$ and the resulting payoffs are integrable.
\end{proof}

\begin{proof}[Proof of \autoref{prop:optimality_cap_allocation}]
Consider the adjusted positive-output payoff $a\mapsto (1-\hat{\kappa})[b(a)-c]+v(a)$. By the curvature condition and the choice of $\hat{\kappa}$, its tangent at $\underline{a}$ passes through the origin. Following this tangent on $[0,\underline{a}]$ and the adjusted payoff above $\underline{a}$ gives a concave upper bound which is tight at shutdown and at every positive output used by the cap. Its slope at shutdown is $(1-\hat{\kappa})b'(\underline{a})+v'(\underline{a})$, and at any output $a\ge \underline{a}$ it is $(1-\hat{\kappa})b'(a)+v'(a)$; these are the slopes appearing in \eqref{eq:monopoly_measure}.

Since
\[ u_P(a,s)-\hat{\kappa}u_A(a,s) = (1-\hat{\kappa})s a + (1-\hat{\kappa})[b(a)-c\mathbf{1}_{a>0}] + v(a), \]
this bound is the hypothesis of \autoref{thm:main_result_nonlinear} in the form discussed in \autoref{sec:general_payoffs}, with $\hat{\kappa}$ in the role of $\kappa$. Hence, for any feasible mechanism $\tilde{m}$ with indirect utility $\tU$,
\[ \int_S \big[u_P(\tilde{m}(s),s) - u_P(q_{\underline{a}}(s),s)\big] \,\mathrm dF(s) \le \int (\tU-U) \,\mathrm d\tilde{\mu}. \]

It remains to show that the last integral is nonpositive. We use the sufficiency comparison of \autoref{thm:outside_option} with $\tilde{\mu}$ in place of $\mu$: for each of the three regions $Q$ below we exhibit a nonnegative measure $\delta_Q$ with $U=h$ on its support such that
\[ \int_Q \tU \,\mathrm d\tilde{\mu} \le \int \tU \,\mathrm d\delta_Q \le \int U \,\mathrm d\delta_Q = \int_Q U \,\mathrm d\tilde{\mu}, \]
where the second inequality uses $\tU\le h$ together with the nonnegativity of $\delta_Q$.

On $[s_2,\infty)$ take $\delta_Q=\tilde{\mu}|_Q$, which is nonnegative by (i) and lives where $U=h$; the first inequality is then an equality.

On $(s_1,s_2)$ take $\delta_Q$ to be the point mass at $s_2$ of mass $\tilde{\mu}((s_1,s_2))$. This mass is nonnegative, being the left derivative at $s_2$ of the expression in (ii), which is nonpositive and vanishes there, and $U(s_2)=h(s_2)$. Testing the functions $(s-x)^+$ yields precisely the inequalities in (ii), and the equality at $s_2$ matches the first moments, so $\tilde{\mu}|_Q\le_{cx}\delta_Q$ and the first inequality holds. Since $U$ is affine on this interval and convex order preserves the integrals of affine functions, the final equality holds as well.

On the shutdown region take $\delta_Q=0$; the final equality holds because $U$ vanishes there. Any feasible $\tU$ is nonnegative, convex, and vanishes on $(-\infty,s_0]$, so on this region it is a limit of positive linear combinations of the functions $(x-s)^+$ with $s\ge s_0$ \citep[p. 183]{shaked2007stochastic}, and condition (iii) therefore gives $\int_{(-\infty,s_1]} \tU \,\mathrm d\tilde{\mu}\le 0$.

Adding the three comparisons shows that the cap allocation is optimal.
\end{proof}

\subsection{The uniform-square example}\label{app:uniform}

We first record a fact about signed barycenters that makes it sufficient to check the first-moment conditions instead of the full convex-order conditions.

\paragraph{From barycenters to convex order.}
Every pool constructed below has the same shape: all of its negative virtual mass sits at a single point, and all of its positive mass lies on the segment between that point and the assigned action $a$. For a pool $Q$ of this shape the signed-barycenter condition already delivers the full convex-order comparison required by \autoref{thm:main_result},
\begin{align}\label{eq:uniform_barycenter_test}
 \int(s-a)\,\mathrm d\mu|_Q(s)=0
 \quad\Longrightarrow\quad
 \mu|_Q(Q)\geq0
 \quad\text{and}\quad
 \mu|_Q\leq_{cx}\delta_Q,
\end{align}
with $\delta_Q$ the point mass at $a$ of mass $\mu|_Q(Q)$, as in the main text. Moreover, $\mu|_Q(Q)>0$ unless all of the positive mass sits at the negatively weighted point.

\paragraph{Small bias: the conditional measures.}
Fix $0<\beta<1/2$ and retain the parametrization $a(t)$ of the main text, writing $q(t):=a_2'(t)/a_1'(t)$ for the slope of the boundary; the constructed branch has $a_1'(t)>0$ and $-1\leq q(t)<0$. Orthogonality \eqref{eq:uniform_normal} gives
\begin{align}\label{eq:uniform_tangent}
 a_1(t)=t+\bigl(1-a_2(t)\bigr)q(t),
 \qquad
 a_1'(t)=\frac{1+\bigl(1-a_2(t)\bigr)q'(t)}{1+q(t)^2}.
\end{align}
The part of the pool labelled $t$ inside $S$ is the segment of the normal through $a(t)$ given by
\[
 Q_t\cap S=\bigl\{\bigl(t+(1-y)q(t),\,y\bigr):\ \underline y(t)\leq y\leq1\bigr\},
 \qquad
 \underline y(t):=\max\bigl\{a_2(t),\,y_0(t)\bigr\},
\]
where $y_0(t):=1+t/q(t)$ is the height at which it meets the vertical axis. By \eqref{eq:uniform_tangent}, $y_0(t)>a_2(t)$ exactly when $a_1(t)<0$, and the types then stop at the west edge rather than at the assigned action. At height $y$ neighboring pools are $\varepsilon+(1-y)[q(t+\varepsilon)-q(t)]$ apart, so the Jacobian is $1+(1-y)q'(t)$, verified positive below; and $y_0'(t)=[1+(1-y_0(t))q'(t)]/q(t)<0$, so the west edge contributes positive mass $-\beta y_0'(t)$ per unit of $t$. Conditional measures are written against the pool label $t$; this differs from the normalization in \autoref{thm:main_result} only by a positive rescaling, which leaves the convex-order comparisons unchanged. Writing the conditional measure in the height coordinate and integrating against $\mathrm dt$ therefore gives
\begin{align}\label{eq:uniform_conditional}
 \mathrm d\mu|_{Q_t}(y)
 =\bigl[1+(1-y)q'(t)\bigr]\mathbf{1}_{[\underline y(t),1]}(y)\,\mathrm dy
 -\beta y_0'(t)\,\mathbf{1}_{\{y_0(t)>a_2(t)\}}\,\delta_{y_0(t)}(\mathrm dy)
 -\beta\,\delta_1(\mathrm dy),
\end{align}
the limit behind the thin band of the main text.

A type $s\in Q_t$ satisfies $s-a(t)=\bigl(s_2-a_2(t)\bigr)\bigl(-q(t),1\bigr)$, so all of this mass lies on one ray from $a(t)$ and the signed-barycenter condition is the scalar equation
\begin{align}\label{eq:uniform_moment}
 \int\bigl(y-a_2(t)\bigr)\,\mathrm d\mu|_{Q_t}(y)=0 .
\end{align}

For a pool inside the square, $\underline y(t)=a_2(t)$, the atom at $y_0(t)$ is absent and \eqref{eq:uniform_moment} reads
\[
 \frac{\bigl(1-a_2(t)\bigr)^2}{2}+\frac{q'(t)\bigl(1-a_2(t)\bigr)^3}{6}
 =\beta\bigl(1-a_2(t)\bigr),
\]
which is precisely the differential equation \eqref{eq:uniform_boundary_ode} of the main text. Solving it for $q'(t)$ and using \eqref{eq:uniform_tangent} gives the interior system
\begin{align}\label{eq:uniform_ode}
 q'(t)=\frac{3\bigl(2\beta-1+a_2(t)\bigr)}{\bigl(1-a_2(t)\bigr)^2},
 \qquad
 a_2'(t)=\frac{2q(t)\bigl(3\beta-1+a_2(t)\bigr)}{\bigl(1-a_2(t)\bigr)\bigl(1+q(t)^2\bigr)}.
\end{align}
For a pool meeting the west edge, $\underline y(t)=y_0(t)$ and the atom there enters as well. Carrying out the integration in \eqref{eq:uniform_moment} and substituting $y_0'(t)$ leaves a linear equation in $q'(t)$, so that
\begin{align}\label{eq:uniform_exterior_ode}
 q'(t)=\frac{\displaystyle
          \beta\bigl(1-a_2(t)\bigr)-\int_{y_0(t)}^1\bigl(y-a_2(t)\bigr)\,\mathrm dy
                         +\frac{\beta}{q(t)}\bigl(y_0(t)-a_2(t)\bigr)}
        {\displaystyle
          \int_{y_0(t)}^1\bigl(y-a_2(t)\bigr)(1-y)\,\mathrm dy
                         -\frac{\beta}{q(t)}\bigl(1-y_0(t)\bigr)\bigl(y_0(t)-a_2(t)\bigr)},
\end{align}
together with \eqref{eq:uniform_tangent}. Its denominator is strictly positive for $a_2(t)\leq y_0(t)<1$ and $q(t)<0$, and at $y_0(t)=a_2(t)$ it reduces to \eqref{eq:uniform_ode}, so the two systems agree at the axis crossing.

\paragraph{The exterior branch.}
Choose a starting depth $\ell\in(0,\beta)$ and impose
\begin{align}\label{eq:uniform_initial_data}
 a(0)=(-\ell,1-\ell),\qquad q(0)=-1.
\end{align}
At $t=0$ the pool collapses to the northwest corner, where the west- and north-edge weights $-\beta/q$ and $-\beta$ cancel exactly when $q(0)=-1$. Equation \eqref{eq:uniform_exterior_ode} has the form $0/0$ there, but substituting $q(t)=-1+t\,z(t)$ and cancelling the common factor $t$ turns it into a regular singular point, so standard arguments give a unique local solution, depending continuously on $\ell$, with
\begin{align}\label{eq:uniform_initial_derivatives}
 q'(0)=\frac{\beta-\ell}{2\beta\ell}>0,
 \qquad
 a_2'(0)=-\frac{3\beta-\ell}{4\beta}.
\end{align}

The slope increases until the boundary reaches the vertical axis.\footnote{Let $N$ be the numerator of \eqref{eq:uniform_exterior_ode}, so that $q'$ has the sign of $N$. At a first zero of $q'$ we would have $-1<q<0$, and $N=0$ would force $1+y_0-2a_2<2\beta$; since $a_1'=1/(1+q^2)$ and $y_0'=1/q$ there, differentiating gives $N'=[q^2(1-a_2-\beta)+(y_0-a_2)+\beta/q]/[q(1+q^2)]>0$, because $q(1+q^2)<0$ while $q^2(1-a_2)+(y_0-a_2)\leq\frac{1+q^2}{2}(1+y_0-2a_2)<\beta(1+q^2)<\beta(q^2-1/q)$. So $N$ cannot reach zero from above.} Hence $-1<q<0$, $a_1'>1/2$ and $a_2'<0$; as $\mathrm da_2/\mathrm da_1=q>-1$ the depth stays below $2\ell$, and $a_1<0$ forces $q<-t/(1-a_2)$, so the slope cannot reach zero first. The branch therefore reaches $a_1=0$ in finite time without a singularity, at a height $a_{2,c}$ and slope $q_c$ with
\begin{align}\label{eq:uniform_axis_bounds}
 \ell<1-a_{2,c}<2\ell<2\beta,
 \qquad -1<q_c<0.
\end{align}

\paragraph{Matching on the diagonal.}
Continue with the interior system \eqref{eq:uniform_ode}; dividing its two equations and integrating gives the first integral
\begin{align}\label{eq:uniform_invariant}
 \bigl(1+q(t)^2\bigr)\bigl(1-a_2(t)\bigr)^2\bigl(3\beta-1+a_2(t)\bigr)=C,\qquad
 C=(1+q_c^2)(1-a_{2,c})^2(3\beta-1+a_{2,c}),
\end{align}
which depends continuously on $\ell$. By \eqref{eq:uniform_axis_bounds}, $C\to0$ as $\ell\downarrow0$ and $C<8\beta^3$ throughout, while as $\ell\uparrow\beta$ the exterior solution converges to $q\equiv-1$, $a_2(t)=1-\beta-t/2$, which reaches the axis at $t=2\beta$, so that $C\to8\beta^3$.

Consider starting values with $4\beta^3<C<8\beta^3$, so that $q=-[C/((1-a_2)^2(3\beta-1+a_2))-1]^{1/2}$ on the interior branch. The slope increases until $a_2=1-2\beta$ and then decreases, returning to $-1$ at the unique height $a_{2,+}\in(1-3\beta,1-2\beta)$ with $2(1-a_{2,+})^2(3\beta-1+a_{2,+})=C$. Since $\mathrm da_1/\mathrm da_2=1/q$ and $a_1=0$ at the axis crossing, the north-edge label at this return is
\begin{align}\label{eq:uniform_matching}
 T(\ell)=1-a_{2,+}
 +\int_{a_{2,+}}^{a_{2,c}}
  \left(\frac{C}{(1-y)^2(3\beta-1+y)}-1\right)^{-1/2}\,\mathrm dy,
\end{align}
and we select $\ell$ so that $T(\ell)=1$. On the connected component of $\{\ell\in(0,\beta):C(\ell)>4\beta^3\}$ whose closure contains $\beta$, the function $T$ is continuous. As $\ell\uparrow\beta$ the crossing and return heights both tend to $1-2\beta$, so $T(\ell)\to2\beta<1$. At its other endpoint $C\downarrow4\beta^3$, and since $(1-y)^2(3\beta-1+y)$ has a strict quadratic maximum $4\beta^3$ at $y=1-2\beta$, the integrand increases there to a limit with a nonintegrable singularity, so $T(\ell)\to\infty$. The intermediate value theorem gives the required starting value, for which $q(1)=-1$ and $a(1)=(a_{2,+},a_{2,+})$; by \eqref{eq:uniform_matching}, $a_{2,+}>0$, so the contact lies strictly inside the square and the boundary joins its reflection smoothly on the diagonal.

\paragraph{The menu and the agent's incentives.}
Orthogonality is only a first-order condition along each normal. It remains to show that the pools partition the type space outside the full-discretion region, so that $U$ is well defined, and that $U$ is convex, so that each type's assigned action is optimal among all actions in the menu.
For the selected solution $a_1'>0$, $a_2'<0$ and $-1\leq q<0$, and the Jacobian $1+(1-y)q'$ is positive on every pool: on the exterior branch because $q'>0$, and on the interior branch because it is affine in $y$, equals one at $y=1$ and equals $2(3\beta-1+a_2)/(1-a_2)>0$ at $y=a_2$. At each height the coordinate $t+(1-y)q(t)$ is therefore strictly increasing in $t$, and $\underline y(t)$ decreases with $t$, so the upper pools do not cross and, together with their reflections, fill the part of $S$ outside the full-discretion region, the closed southwest region bounded by the coordinate axes, the interior portion of $a(t)$ and its reflection; $a_{2,+}>0$ keeps these segments inside the square. The menu is this region together with the two exterior arcs.

Set $U=h$ on the full-discretion region and $U(s)=a(t)\cdot s-\tfrac12\norm{a(t)}^2$ for $s\in Q_t$, with the reflected definition below the diagonal. Orthogonality gives $\nabla U(s)=a(t)$ and
\begin{align}\label{eq:uniform_hessian}
 D^2U(s)=\frac{a_1'(t)}{1+(1-s_2)q'(t)}
 \begin{pmatrix}1&q(t)\\q(t)&q(t)^2\end{pmatrix}\succeq0,
\end{align}
while $D^2U=I$ on the full-discretion region. Utilities and gradients agree across the boundary and the diagonal and these derivatives are bounded, so the piecewise smooth gradient is Lipschitz and $U$ is convex on $S$. Every offered action, including the limiting ones at the ends of the exterior arcs, generates a supporting payoff plane, so the menu implements the stated assignments, and the supremum of those planes is convex, lies below $h$ on $\R^2$ and equals $h$ at every offered action.

\paragraph{Optimality and the necessity of nonconvexity.}
On each nontrivial pool the positive part of $\mu|_{Q_t}$ lies between the assigned action and the north- or east-edge endpoint carrying its negative mass, and \eqref{eq:uniform_moment} is its signed-barycenter condition, so \eqref{eq:uniform_barycenter_test} gives $\mu|_{Q_t}(Q_t)>0$ and $\mu|_{Q_t}\leq_{cx}\delta_{Q_t}$. Hence every feasible indirect utility $V$ satisfies
\[ \int V\,\mathrm d\mu|_{Q_t}\leq\int V\,\mathrm d\delta_{Q_t}\leq\int U\,\mathrm d\delta_{Q_t}=\int U\,\mathrm d\mu|_{Q_t}, \]
while on the remaining types the virtual measure is nonnegative and $U=h$. Integrating proves optimality among all feasible mechanisms, including stochastic ones.

The same comparison gives uniqueness: equality forces any other optimum $V$ to satisfy $V(a(t))=h(a(t))$ for almost every pool and to be affine on the segment from $a(t)$ to its negatively weighted endpoint, and since $V\leq h$ that affine restriction is the tangent plane of $h$ at $a(t)$, hence $U$; on the full-discretion region equality requires $V=h$. So $V=U$ on $S$ by continuity, and the quadratic payoff identity forces every optimal mechanism to assign a lottery with zero variance to almost every type.

Finally, $1-a_{2,c}<2\beta$ by \eqref{eq:uniform_axis_bounds}, so $q'>0$ on an interior arc just after the boundary crosses the vertical axis, where $\mathrm d^2a_2/\mathrm da_1^2=q'/a_1'>0$. Take two nearby contact points on this arc, both strictly inside $S$; their midpoint lies above the boundary, where $U<h$. A closed convex delegation set implementing an optimum would contain both contact points, since the unique optimal utility equals $h$ there, hence also their midpoint, and would give that type full discretion, contradicting $U<h$. No closed convex delegation set is therefore optimal for $0<\beta<1/2$.

\end{document}